\documentclass[12pt,technote, onecolumn, draft]{IEEEtran}
\usepackage{latexsym}
\usepackage{mathrsfs}
\usepackage{multirow}
\usepackage{amsfonts,amssymb,amsmath,amsthm,bm}
\usepackage{color}
\usepackage{cite}
\newtheorem{theorem}{Theorem}
\newtheorem{example}{Example}

\newtheorem{remark}{Remark}
\newtheorem{corollary}{Corollary}
\newtheorem{lemma}{Lemma}
\newtheorem{proposition}{Proposition}

\begin{document}

\title{Some New Constructions of Generalized Plateaued Functions$^{\dag}$}
\author{Jiaxin Wang, Fang-Wei Fu
\IEEEcompsocitemizethanks{\IEEEcompsocthanksitem Jiaxin Wang and Fang-Wei Fu are with Chern Institute of Mathematics and LPMC, Nankai University, Tianjin 300071, P.R.China, Emails: wjiaxin@mail.nankai.edu.cn, fwfu@nankai.edu.cn.
}
\thanks{$^\dag$This research is supported by the National Key Research and Development Program of China (Grant No. 2018YFA0704703), the National Natural Science Foundation of China (Grant Nos. 12141108 and 61971243), the Natural Science Foundation of Tianjin (20JCZDJC00610), the Fundamental Research Funds for the Central Universities of China (Nankai University), and the Nankai Zhide Foundation.}}

\maketitle

\begin{abstract}
  Plateaued functions as an extension of bent functions play a significant role in cryptography, coding theory, sequences
  and combinatorics. In \cite{Mesnager9}, Mesnager \emph{et al.} introduced generalized plateaued functions in order to study plateaued functions in the general context of generalized $p$-ary functions. In this paper, we focus on the constructions of generalized $p$-ary $s$-plateaued functions from $V_{n}$ to $\mathbb{Z}_{p^k}$, where $V_{n}$ is an $n$-dimensional vector space over $\mathbb{F}_{p}$, $p$ is a prime, $k\geq 1$ and $n+s$ is even when $p=2$. In particular, when $k=1$, the constructions in this paper are applicable for plateaued functions. Firstly, inspired by the work of Hod\v{z}i\'{c} \emph{et al}. \cite{Hodzic3} for Boolean plateaued functions, we characterize generalized plateaued functions with affine Walsh supports and provide constructions of generalized plateaued functions with (non)-affine Walsh supports by spectral method. When $p=2, k=1$, our constructions of Boolean plateaued functions with (non)-affine Walsh supports provide an answer to the Open Problem 2 proposed in \cite{Hodzic3}. Secondly, based on what we called generalized indirect sum, we give constructions of generalized plateaued functions, which are also applicable for (non)-weakly regular generalized bent functions. In particular, we show that the canonical way to construct Generalized Maiorana-McFarland bent functions can be obtained by the generalized indirect sum and we illustrate that the generalized indirect sum can be used to construct bent functions not in the completed Generalized Maiorana-McFarland class. Based on the generalized indirect sum, we also give constructions of plateaued functions in the subclass \emph{WRP} of the class of weakly regular plateaued functions and vectorial plateaued functions. In the end, we discuss the constructions of pairwise disjoint spectra generalized plateaued functions with (non)-affine Walsh supports and we present a construction of generalized bent functions by using pairwise disjoint spectra generalized plateaued functions as building blocks.
\end{abstract}


%-------------------------------------------------------------------------------
\section{Introduction}
\label{intro}
Boolean bent functions introduced by Rothaus \cite{Rothaus} play an important role in cryptography, coding theory, sequences and combinatorics. Kumar \emph{et al}. \cite{Kumar} generalized Boolean bent functions to bent functions over finite fields of odd characteristic. Due to the importance of bent functions, they have been extensively studied. There is an exhaustive survey \cite{Carlet5} and books \cite{Carlet2,Mesnager1} on bent functions. Recently, Mesnager \emph{et al.} \cite{Mesnager10} introduced generalized bent functions from $V_{n}$ to $\mathbb{Z}_{p^k}$, where $V_{n}$ is an $n$-dimensional vector space over $\mathbb{F}_{p}$, $p$ is a prime. For more characterizations and constructions of generalized bent functions from $V_{n}$ to $\mathbb{Z}_{p^k}$, we refer to \cite{Hodzic1,Hodzic2,Martinsen1,Martinsen2,Meidl,Mesnager2,Mesnager10,Qi,Stanica,Tang}.

In \cite{Carlet4}, Carlet introduced Boolean partially bent functions which is an extension of Boolean bent functions. As an extension of Boolean partially bent functions, Zheng and Zhang \cite{Zheng} introduced Boolean plateaued functions. Surveys on Boolean plateaued functions can be found in \cite{Carlet1,Carlet2,Mesnager1}. The notion of Boolean partially bent functions and Boolean plateaued functions have been generalized to $p$-ary partially bent functions and $p$-ary plateaued functions for any odd prime $p$ (see \cite{Cesmelioglu1,Cesmelioglu2}). Then they have been studied in \cite{Cesmelioglu1,Cesmelioglu2,Hyun,Mesnager4,Mesnager5,Potapov}. In \cite{Hyun}, Hyun \emph{et al}. searched for explicit criteria for constructing $p$-ary plateaued functions. More specifically, for $p$-ary $s$-plateaued functions, they derived an explicit form for the Walsh transform, obtained an upper bound on the degree and provided explicit criteria for the existence. In \cite{Mesnager4,Mesnager5}, Mesnager \emph{et al}. presented characterizations of $p$-ary plateaued functions in terms of the second-order derivatives and the moments of Walsh transform, which allow us a better understanding of the structure of $p$-ary plateaued functions. Apart from the desirable cryptographic properties, plateaued functions play a significant role in coding theory, sequences and combinatorics (see e.g. \cite{Serdar,Mesnager3,Mesnager7,Mesnager8,Oktay}). In \cite{Mesnager9}, Mesnager \emph{et al.} introduced generalized plateaued functions from $V_{n}$ to $\mathbb{Z}_{p^k}$ in order to study plateaued functions in the general context of generalized $p$-ary functions. As far as we know, there are only a few papers on generalized plateaued functions \cite{Riera,Mesnager6,Mesnager9} up to now. We review the main contributions for generalized plateaued functions given in these papers. In \cite{Mesnager9}, first of all, the authors gave an explicit form for the Walsh transform of generalized plateaued functions. They then investigated the relations between generalized plateaued functions and plateaued functions by the decomposition of generalized plateaued functions. In particular, they used admissible plateaued functions to characterize generalized plateaued functions by means of their components. Finally, they provided for the first time two constructions of generalized Boolean plateaued functions. In \cite{Riera}, for generalized Boolean plateaued functions, the authors provided two constructions and characterized them in terms of the second-order derivatives and the fourth moment of Walsh transform. In \cite{Mesnager6}, a special class of generalized plateaued functions called $\mathbb{Z}_{2^k}$-plateaued functions was studied in terms of so called $(c, s)$-plateaued functions. In particular, the authors gave characterizations of $(2^t, s)$-plateaued functions in terms of the second-order derivatives and the fourth moment of Walsh transform, which generalize the results given in \cite{Riera}. They pointed out that even though the paper \cite{Mesnager6} only stated the results for characteristic $2$, similar results can be obtained for odd characteristic. For generalized $p$-ary plateaued functions, the constructions in \cite{Riera,Mesnager9} are for $p=2$ and there are lacks of constructions with $p$ taking any prime. The main contribution of this paper (which will be introduced below) is to provide constructions of generalized $p$-ary plateaued functions for any prime $p$.

Recently, Hod\v{z}i\'{c} \emph{et al}. \cite{Hodzic3} designed Boolean plateaued functions in spectral domain. Designing plateaued functions in spectral domain is based on the fact that any function and its Walsh transform are mutually determined. In this paper, we focus on the constructions of generalized $s$-plateaued functions from $V_{n}$ to $\mathbb{Z}_{p^k}$, where $V_{n}$ is an $n$-dimensional vector space over $\mathbb{F}_{p}$, $p$ is a prime, $k\geq 1$ and $n+s$ is even when $p=2$. In particular, when $k=1$, the constructions in this paper are applicable for plateaued functions. Firstly, inspired by the work of Hod\v{z}i\'{c} \emph{et al}. \cite{Hodzic3}, we characterize generalized plateaued functions with affine Walsh supports and provide constructions of generalized plateaued functions with (non)-affine Walsh supports in spectral domain. As pointed out in \cite{Hodzic3}, for the constructions in spectral domain given in \cite{Hodzic3}, the Walsh supports of Boolean $s$-plateaued functions in $n$ variables, when written as matrices, contain at least $n-s$ columns corresponding to affine functions on $\mathbb{F}_{2}^{n-s}$. They proposed an open problem (Open Problem 2) to provide constructions of Boolean $s$-plateaued functions in $n$ variables whose Walsh supports, when written as matrices, contain strictly less than $n-s$ columns corresponding to affine functions. In our constructions of generalized $s$-plateaued functions with non-affine Walsh supports, the Walsh supports, when written as matrices, can contain strictly less than $n-s$ columns corresponding to affine functions. When $p=2, k=1$, these constructions provide an answer to Open Problem 2 proposed in \cite{Hodzic3}. Secondly, based on what we called generalized indirect sum, we provide constructions of generalized plateaued functions, which are also applicable for (non)-weakly regular generalized bent functions. In particular, we show that the canonical way to construct Generalized Maiorana-McFarland bent functions can be obtained by the generalized indirect sum and we illustrate that the generalized indirect sum can be used to construct bent functions not in the completed Generalized Maiorana-McFarland class. Based on the generalized indirect sum, we also give constructions of plateaued functions in the subclass \emph{WRP} of the class of weakly regular plateaued functions and vectorial plateaued functions. In the end, we discuss the constructions of pairwise disjoint spectra generalized plateaued functions with (non)-affine Walsh supports and we present a construction of generalized bent functions by using pairwise disjoint spectra generalized plateaued functions as building blocks.

The rest of the paper is organized as follows. In Section II, we introduce the needed definitions and results related to generalized plateaued functions. In Section III, based on the principle of designing generalized plateaued functions in spectral domain, we characterize generalized plateaued functions with affine Walsh supports and provide constructions of generalized plateaued functions with (non)-affine Walsh supports. In Section IV, based on what we called generalized indirect sum, we give constructions of generalized plateaued functions, which are also applicable for generalized bent functions. In Section V, we discuss the constructions of pairwise disjoint spectra generalized plateaued functions with (non)-affine Walsh supports and we present a construction of generalized bent functions by using pairwise disjoint spectra generalized plateaued functions as building blocks. In Section VI, we make a conclusion.
\section{Preliminaries}
\label{sec:1}
Throughout this paper, let $\mathbb{Z}_{p^k}$ be the ring of integers modulo $p^k$, $\mathbb{F}_{p}^{n}$ be the vector space of the $n$-tuples over $\mathbb{F}_{p}$, $\mathbb{F}_{p^n}$ be the finite field with $p^n$ elements and $V_{n}$ be an $n$-dimensional vector space over $\mathbb{F}_{p}$, where $p$ is a prime and $k, n$ are positive integers. The classical representations of $V_{n}$ are $\mathbb{F}_{p}^{n}$ and $\mathbb{F}_{p^n}$. For $a, b \in V_{n}$, let $\langle a, b \rangle$ denote a (non-degenerate) inner product of $V_{n}$. When $a=(a_{1}, \dots, a_{n}), b=(b_{1}, \dots, b_{n})\in \mathbb{F}_{p}^{n}$, let $\langle a, b \rangle=a \cdot b=\sum_{i=1}^{n}a_{i}b_{i}$. When $a, b \in \mathbb{F}_{p^n}$, let $\langle a, b \rangle=Tr_{1}^{n}(ab)$, where $Tr_{1}^{n}(\cdot)$ is the absolute trace function. When $V_{n}=V_{n_{1}}\times \dots \times V_{n_{s}}\ (n=\sum_{i=1}^{s}n_{i})$, let $\langle a, b\rangle =\sum_{i=1}^{s}\langle a_{i}, b_{i}\rangle $, where $a=(a_{1}, \dots, a_{s}), b=(b_{1}, \dots, b_{s})\in V_{n}$.

A function $f$ from $V_{n}$ to $\mathbb{Z}_{p^k}$ is called a generalized $p$-ary function, or simply $p$-ary function when $k=1$. A $p$-ary function $L: V_{n}\rightarrow \mathbb{F}_{p}$ is called a linear function if $L(ax+by)=aL(x)+bL(y)$ for any $a, b \in \mathbb{F}_{p}$ and $x, y \in V_{n}$. All linear functions from $V_{n}$ to $\mathbb{F}_{p}$ form an $n$-dimensional linear space $\mathcal{L}_{n}$ and $\{\langle \alpha_{i}, x\rangle, 1\leq i\leq n\}$ is a basis of $\mathcal{L}_{n}$, where $\{\alpha_{i}, 1\leq i \leq n\}$ is a basis of $V_{n}$. If a $p$-ary function $A: V_{n}\rightarrow \mathbb{F}_{p}$ is the sum of a linear function and a constant, then $A$ is called an affine function.

The Walsh transform of a generalized $p$-ary function $f: V_{n} \rightarrow \mathbb{Z}_{p^k}$ is the complex valued function $W_{f}$ on $V_{n}$ defined as
 \begin{equation}\label{1}
  W_{f}(a)=\sum_{x\in V_{n}}\zeta_{p^k}^{f(x)}\zeta_{p}^{-\langle a, x \rangle}, \ a \in V_{n},
 \end{equation}
where for any positive integer $q$, $\zeta_{q}=e^{\frac{2\pi \sqrt{-1}}{q}}$ is the complex primitive $q$-th root of unity. The generalized $p$-ary function $f$ can be recovered by the inverse transform
 \begin{equation}\label{2}
  \zeta_{p^k}^{f(x)}=\frac{1}{p^n}\sum_{a\in V_{n}}W_{f}(a)\zeta_{p}^{\langle a, x\rangle}, \ x \in V_{n}.
 \end{equation}

The multiset $\{W_{f}(a), a \in V_{n}\}$ is called the Walsh spectrum of $f$. The set $S_{f}=\{a\in V_{n}: W_{f}(a)\neq 0\}$ is called the Walsh support of $f$. Functions $f_{1}, \dots, f_{m}$ are called pairwise disjoint spectra functions if $S_{f_{i}}\cap S_{f_{j}}=\emptyset $ for any $i\neq j$.

A generalized $p$-ary function $f: V_{n}\rightarrow \mathbb{Z}_{p^k}$ is called a generalized $p$-ary $s$-plateaued function, or simply $p$-ary $s$-plateaued function when $k=1$ if $|W_{f}(a)|=p^{\frac{n+s}{2}}$ or $0$ for any $a\in V_{n}$. If $s=0$, the generalized $p$-ary $0$-plateaued function $f$ is just the generalized $p$-ary bent function and $S_{f}=V_{n}$. When $p=2, k=1$, if $f: V_{n}\rightarrow \mathbb{F}_{2}$ is an $s$-plateaued function, then $n+s$ is even.

For generalized $s$-plateaued functions $f:V_{n}\rightarrow \mathbb{Z}_{p^k}$, there is a basic property: $|S_{f}|=p^{n-s}$, which is obtained by the Parseval identity $\sum_{x \in V_{n}}|W_{f}(x)|^{2}=p^{2n}$. In \cite{Mesnager9}, Mesnager \emph{et al}. have shown that the Walsh transform of a generalized $p$-ary $s$-plateaued function $f: V_{n}\rightarrow \mathbb{Z}_{p^k}$ satisfies that for any $a\in S_{f}$,
when $p=2$ and $n+s$ is even, $W_{f}(a)=2^{\frac{n+s}{2}}\zeta_{2^k}^{f^{*}(a)}$, and when $p$ is an odd prime,
 \begin{equation*}
  W_{f}(a)=\left\{\begin{array}{cc}
                    \pm p^{\frac{n+s}{2}}\zeta_{p^k}^{f^{*}(a)}& \ \text{if} \ n+s \ \text{is even or} \ p\equiv 1 \ (mod\ 4), \\
                    \pm \sqrt{-1} p^{\frac{n+s}{2}}\zeta_{p^k}^{f^{*}(a)}& \ \ \text{if} \ n+s \ \text{is odd and} \ p\equiv 3 \ (mod\ 4),
                  \end{array}\right.
 \end{equation*}
where $f^{*}$ is a function from $S_{f}$ to $\mathbb{Z}_{p^k}$. We call $f^{*}$ the dual of $f$.

In the sequel, if $f: V_{n}\rightarrow \mathbb{Z}_{p^k}$ is a generalized $s$-plateaued function with dual $f^{*}$, define function $\mu_{f}$ as
\begin{equation}\label{3}
  \mu_{f}(a)=p^{-\frac{n+s}{2}}\zeta_{p^k}^{-f^{*}(a)}W_{f}(a), a \in S_{f}.
\end{equation}
If $p\equiv 1 \ (mod  \ 4)$ or $p \equiv 3 \ (mod \ 4)$ and $n+s$ is even, then $\mu_{f}$ is a function from $S_{f}$ to $\{\pm 1\}$. If $p \equiv 3 \ (mod \ 4)$ and $n+s$ is odd, then $\mu_{f}$ is a function from $S_{f}$ to $\{\pm \sqrt{-1}\}$. If $p=2$ and $n+s$ is even, then $\mu_{f}(x)=1, x \in S_{f}$. For a generalized bent function $f: V_{n}\rightarrow \mathbb{Z}_{p^k}$, that is, generalized $0$-plateaued function, if $\mu_{f}$ is a constant function, then $f$ is called weakly regular, otherwise $f$ is called non-weakly regular. In particular, if $\mu_{f}(x)=1, x \in V_{n}$, $f$ is called regular. In \cite{Mesnager3}, Mesnager \emph{et al}. introduced the notion of (non)-weakly regular plateaued functions. For an $s$-plateaued function $f: V_{n}\rightarrow \mathbb{F}_{p}$, if $\mu_{f}$ is a constant function, then $f$ is called weakly regular, otherwise $f$ is called non-weakly regular. In particular, if $\mu_{f}(x)=1, x \in S_{f}$, $f$ is called regular.

We recall some well-known (generalized) bent functions.
\begin{itemize}
  \item When $V_{n}=\mathbb{F}_{p^n}\times \mathbb{F}_{p^n}$, let $f: V_{n} \rightarrow \mathbb{Z}_{p^k}$ be defined as
  \begin{equation}\label{4}
    f(x, y)=p^{k-1}Tr_{1}^{n}(\alpha x \pi(y))+g(y),
  \end{equation}
  where $\alpha \in \mathbb{F}_{p^n}^{*}$, $\pi$ is a permutation over $\mathbb{F}_{p^n}$ and $g: \mathbb{F}_{p^n}\rightarrow \mathbb{Z}_{p^k}$ is an arbitrary function.
  When $V_{n}=\mathbb{F}_{p}^{n} \times \mathbb{F}_{p}^{n}$, let $f: V_{n} \rightarrow \mathbb{Z}_{p^k}$ be defined as
  \begin{equation}\label{5}
    f(x, y)=p^{k-1} x \cdot \pi(y)+g(y),
  \end{equation}
  where $\pi$ is a permutation over $\mathbb{F}_{p}^{n}$ and $g: \mathbb{F}_{p}^{n}\rightarrow \mathbb{Z}_{p^k}$ is an arbitrary function.
  Then $f$ defined by (4), respectively (5), is called a Maiorana-McFarland generalized bent function with
  \begin{equation}\label{6}
    W_{f}(x, y)=p^{n}\zeta_{p^k}^{p^{k-1}Tr_{1}^{n}(-\pi^{-1}(\alpha^{-1}x)y)+g(\pi^{-1}(\alpha^{-1}x))},
  \end{equation}
  respectively,
  \begin{equation}\label{7}
    W_{f}(x, y)=p^{n}\zeta_{p^k}^{-p^{k-1}\pi^{-1}(x)\cdot y+g(\pi^{-1}(x))}.
  \end{equation}
  \item Let $f: \mathbb{F}_{p^n} \times \mathbb{F}_{p^n}\rightarrow \mathbb{F}_{p}$ be defined as
  \begin{equation*}
    f(x, y)=Tr_{1}^{n}(\alpha G(xy^{p^n-2})),
  \end{equation*}
  where $\alpha \in \mathbb{F}_{p^n}^{*}$ and $G$ is a permutation over $\mathbb{F}_{p^n}$ with $G(0)=0$. Then $f$ is called a $p$-ary $PS_{ap}$ bent function (which is a generalization of Boolean $PS_{ap}$ bent functions \cite{Dillon}) with
  \begin{equation}\label{8}
    W_{f}(x, y)=p^{n}\zeta_{p}^{Tr_{1}^{n}(\alpha G(-x^{p^n-2}y))}.
  \end{equation}
  The class of $PS_{ap}$ bent functions is a subclass of the famous class of partial spread bent functions. For partial spread bent functions, we refer to \cite{Dillon}, \cite{Lisonek}.
  \item  Let $p$ be an odd prime, and let $\eta$ be the multiplicative quadratic character of $\mathbb{F}_{p^n}$, that is, $\eta(x)=1$ if $x \in \mathbb{F}_{p^n}^{*}$ is a square and $\eta(x)=-1$ if $x \in \mathbb{F}_{p^n}^{*}$ is a non-square. Let $f: \mathbb{F}_{p^n}\rightarrow \mathbb{F}_{p}$ be defined as $f(x)=Tr_{1}^{n}(\alpha x^{2})$, where $\alpha \in \mathbb{F}_{p^n}^{*}$. Then $f$ is a bent function with
      \begin{equation}\label{9}
        W_{f}(a)=(-1)^{n-1}\epsilon^{n}\eta(\alpha)p^{\frac{n}{2}}\zeta_{p}^{Tr_{1}^{n}(-\frac{a^2}{4\alpha})},
      \end{equation}
  where $\epsilon=1$ if $p\equiv 1 \ (mod \ 4)$ and $\epsilon=\sqrt{-1}$ if $p\equiv 3 \ (mod \ 4)$ (see \cite{Helleseth}).
 \end{itemize}

If $f: V_{n}\rightarrow \mathbb{Z}_{p^k}$ is a generalized $n$-plateaued function, then $|S_{f}|=1$ and it is easy to obtain $f(x)=p^{k-1}\langle a, x\rangle +b$ for some $a \in V_{n}, b\in \mathbb{Z}_{p^k}$ by the inverse Walsh transform (2). In this paper, we study generalized $s$-plateaued functions $f: V_{n}\rightarrow \mathbb{Z}_{p^k}$, where $0\leq s < n$, $p$ is prime, $k\geq 1$ and $n+s$ is even when $p=2$.
\section{Constructing generalized plateaued functions by spectral method}
In this section, we provide some constructions of generalized $s$-plateaued functions by spectral method, where $s\geq 1$.

We fix some notation unless otherwise stated. Let $m$ be an arbitrary positive integer. Define the notation of lexicographic order $\prec : a\prec b$ if $\sum_{i=1}^{m}p^{m-i}a_{i}$ $<\sum_{i=1}^{m}p^{m-i}b_{i}$, where $a=(a_{1}, \dots, a_{m}), b=(b_{1}, \dots, b_{m})\in \mathbb{F}_{p}^{m}$. Define
\begin{equation}\label{10}
  v_{i}=\sum_{j=1}^{m}v_{i, j}\alpha_{j}, 0\leq i \leq p^m-1,
\end{equation}
where $\{\alpha_{1}, \dots, \alpha_{m}\}$ is some fixed basis of $V_{m}$ over $\mathbb{F}_{p}$ and $\{(v_{0, 1}, \dots, v_{0, m}), \dots, $ $(v_{p^m-1, 1}, \dots, $ $v_{p^m-1, m})\}$ is the lexicographic order of $\mathbb{F}_{p}^{m}$. When $V_{m}=\mathbb{F}_{p}^{m}$, we let $\alpha_{1}=(1, 0, \dots, 0, 0)\in \mathbb{F}_{p}^{m}, \dots,$ $\alpha_{m}=(0, 0, \dots, 0, 1)\in \mathbb{F}_{p}^{m}$, that is, $\{v_{0}, \dots, v_{p^m-1}\}$ denotes the lexicographic order of $\mathbb{F}_{p}^{m}$. For a $p$-ary function $f: V_{m}\rightarrow \mathbb{F}_{p}$, define its true table
\begin{equation}\label{11}
  T_{f}=\left(f(v_{0}), \dots, f(v_{p^m-1})\right)^{T},
\end{equation}
where $M^{T}$ denotes the transpose of matrix $M$. Let $\delta$ be the Kronecker delta function, that is,
\begin{equation*}
  \delta (i, j)=\left\{\begin{array}{cc}
                 1  & \ \text{if} \ i=j, \\
                 0  & \ \text{if} \ i\neq j.
                 \end{array}\right.
\end{equation*}
\subsection{The principle of designing generalized plateaued functions in spectral domain}
In this subsection, we explain the principle of designing generalized plateaued functions in spectral domain.

Suppose $S \subseteq \mathbb{F}_{p}^{n}$ with size $p^{m}$ is ordered as $S=\{w_{0}, w_{1}, \dots, w_{p^{m}-1}\}$. For any $a\in \mathbb{F}_{p}^{n}$, define $\psi_{a}$ from $V_{m}$ to $\mathbb{F}_{p}$:
\begin{equation}\label{12}
 \psi_{a}(v_{i})=a\cdot w_{i}, 0\leq i \leq p^{m}-1,
 \end{equation}
where $v_{i}$ is defined by (10).

Under notation as above we have the following proposition, which describes the principle of designing generalized plateaued functions in spectral domain. When $p=2, k=1$, the following proposition reduces to i) of Theorem 3.3 of \cite{Hodzic3}.
\begin{proposition}\label{Proposition1}
Let $p$ be a prime, $n, k, s \ (< n)$ be positive integers and $n+s$ be even for $p=2$. Let $S$ be a subset of $\mathbb{F}_{p}^{n}$ with size $p^{n-s}$ and be ordered as $S=\{w_{0}, w_{1}, \dots, w_{p^{n-s}-1}\}$. Let $d$ be a function from $V_{n-s}$ to $\mathbb{Z}_{p^k}$. Let $\mu$ be a function from $V_{n-s}$ to $\{\pm 1 \}$ if $p \equiv 1 \ (mod \ 4)$ or $p \equiv 3 \ (mod \ 4)$ and $n+s$ is even, $\mu$ be a function from $V_{n-s}$ to $\{\pm \sqrt{-1}\}$ if $p \equiv 3 \ (mod \ 4)$ and $n+s$ is odd, and $\mu(x)=1, x \in V_{n-s}$ if $p=2$ and $n+s$ is even. Define the complex valued function $W$ on $\mathbb{F}_{p}^{n}$ as
\begin{equation}\label{13}
  W(a)=p^{\frac{n+s}{2}}\sum_{i=0}^{p^{n-s}-1}\delta(a, w_{i})\mu(v_{i})\zeta_{p^k}^{d(v_{i})}, a \in \mathbb{F}_{p}^{n}.
\end{equation}
Then $W$ is the Walsh transform of a generalized $s$-plateaued function $f: \mathbb{F}_{p}^{n} \rightarrow \mathbb{Z}_{p^k}$ if and only if $(p^{\frac{s-n}{2}}\sum_{x\in V_{n-s}}\mu(x)\zeta_{p^k}^{d(x)+p^{k-1}\psi_{a}(x)})^{p^k}=1$ for any $a\in \mathbb{F}_{p}^{n}$, where $\psi_{a}$ is defined by (12).
\end{proposition}
\begin{proof}
If $W$ is the Walsh transform of a generalized $s$-plateaued function $f: \mathbb{F}_{p}^{n} \rightarrow \mathbb{Z}_{p^k}$, then by the inverse Walsh transform (2) we have
\begin{equation*}
  \begin{split}
    \zeta_{p^k}^{f(a)}&=\frac{1}{p^n}\sum_{i=0}^{p^{n-s}-1}\mu(v_{i})p^{\frac{n+s}{2}}\zeta_{p^k}^{d(v_{i})}\zeta_{p}^{a\cdot w_{i}}\\
       &=p^{\frac{s-n}{2}}\sum_{i=0}^{p^{n-s}-1}\mu(v_{i})\zeta_{p^k}^{d(v_{i})+p^{k-1}\psi_{a}(v_{i})}\\
       &=p^{\frac{s-n}{2}}\sum_{x\in V_{n-s}}\mu(x)\zeta_{p^k}^{d(x)+p^{k-1}\psi_{a}(x)},
  \end{split}
\end{equation*}
thus $(p^{\frac{s-n}{2}}\sum_{x\in V_{n-s}}\mu(x)\zeta_{p^k}^{d(x)+p^{k-1}\psi_{a}(x)})^{p^k}=1$ for any $a\in \mathbb{F}_{p}^{n}$.

Conversely, suppose $(p^{\frac{s-n}{2}}\sum_{x\in V_{n-s}}\mu(x)\zeta_{p^k}^{d(x)+p^{k-1}\psi_{a}(x)})^{p^k}=1$ for any $a\in \mathbb{F}_{p}^{n}$. Then there is a unique generalized $p$-ary function $f: \mathbb{F}_{p}^{n}\rightarrow \mathbb{Z}_{p^k}$ such that $p^{\frac{s-n}{2}}\sum_{x\in V_{n-s}}\mu(x)\zeta_{p^k}^{d(x)+p^{k-1}\psi_{a}(x)}$ $=\zeta_{p^k}^{f(a)}$. The function $W$ is the Walsh transform of $f$. Indeed,
\begin{equation*}
  \begin{split}
    W_{f}(a) & =\sum_{x\in \mathbb{F}_{p}^{n}}(p^{\frac{s-n}{2}}\sum_{y\in V_{n-s}}\mu(y)\zeta_{p^k}^{d(y)+p^{k-1}\psi_{x}(y)})\zeta_{p}^{-a\cdot x}\\
             & =p^{\frac{s-n}{2}}\sum_{x\in\mathbb{F}_{p}^{n}}\sum_{i=0}^{p^{n-s}-1}\mu(v_{i})\zeta_{p^k}^{d(v_{i})+p^{k-1}x\cdot w_{i}}\zeta_{p}^{-a\cdot x}\\
  \end{split}
\end{equation*}
\begin{equation*}
\begin{split}
             &=p^{\frac{s-n}{2}}\sum_{i=0}^{p^{n-s}-1}\mu(v_{i})\zeta_{p^k}^{d(v_{i})}\sum_{x\in \mathbb{F}_{p}^{n}}\zeta_{p}^{(w_{i}-a)\cdot x}.
   \end{split}
\end{equation*}
If $a\notin S=\{w_{0}, w_{1}, \dots,w_{p^{n-s}-1}\}$, then $W_{f}(a)=0$. If $a=w_{i}$ for some $0 \leq i \leq p^{n-s}-1$, then $W_{f}(a)=p^{\frac{n+s}{2}}\mu(v_{i})\zeta_{p^k}^{d(v_{i})}$. Hence, $W_{f}(a)=W(a)$ for any $a \in \mathbb{F}_{p}^{n}$ and $S_{f}=S$, $|W_{f}(a)|=p^{\frac{n+s}{2}}$ for any $a \in S_{f}$, that is, $W$ is the Walsh transform of $f$ and $f$ is a generalized $s$-plateaued function.
\end{proof}

Let $W_{K}$ denote the group of roots of unity of cyclotomic field $K=\mathbb{Q}(\zeta_{p^k})$, then $W_{K}=\{\zeta_{2^k}^{i}: 0 \leq i \leq 2^k-1\}$ if $p=2$ and $W_{K}=\{\pm \zeta_{p^k}^{i}:0\leq i \leq p^k-1\}$ if $p$ is an odd prime. Let $p^{*}=\left(\frac{-1}{p}\right)p$ if $p$ is an odd prime, where $\left(\frac{-1}{p}\right)=(-1)^{\frac{p-1}{2}}$ denotes the Legendre symbol and $p^{*}=2$ if $p=2$. By the knowledge on cyclotomic field $\mathbb{Q}(\zeta_{p^k})$ (see Lemma 24 of \cite{Mesnager9}), $\frac{\alpha}{\sqrt{p^{*}}^{m}}\in W_{K}$ if $\alpha \in \mathbb{Z}[\zeta_{p^k}]$ with $|\alpha|=p^{\frac{m}{2}}$, where $m$ is a positive integer and $m$ is even if $p=2$. Then it is easy to verify that the necessary and sufficient condition in Proposition 1 can be written in the following form.
\begin{proposition}\label{Proposition2}
Keep the same notation as in Proposition 1.

(1) When $p=2$ and $n+s$ is even, the function $W$ defined by (13) is the Walsh transform of a generalized $s$-plateaued function $f: \mathbb{F}_{p}^{n} \rightarrow \mathbb{Z}_{p^k}$ if and only if $|\sum_{x\in V_{n-s}}\zeta_{p^k}^{d(x)+p^{k-1}\psi_{a}(x)}|=p^{\frac{n-s}{2}}$ for any $a\in \mathbb{F}_{p}^{n}$.

(2) When $p$ is an odd prime, the function $W$ defined by (13) is the Walsh transform of a generalized $s$-plateaued function $f: \mathbb{F}_{p}^{n} \rightarrow \mathbb{Z}_{p^k}$ if and only if $|\sum_{x\in V_{n-s}}$ $\mu(x)\zeta_{p^k}^{d(x)+p^{k-1}\psi_{a}(x)}|=p^{\frac{n-s}{2}}$ and $(p^{\frac{s-n}{2}}\sum_{x\in V_{n-s}}$ $\mu(x)\zeta_{p^k}^{d(x)+p^{k-1}\psi_{a}(x)})^{p^k}\neq -1$ for any $a\in \mathbb{F}_{p}^{n}$.
\end{proposition}

By Proposition 1, we obtain the following corollary.

\begin{corollary}\label{Corollary1}
Let $p$ be a prime, $n, k, s \ (< n)$ be positive integers and $n+s$ be even for $p=2$. Let $S$ be a subset of $\mathbb{F}_{p}^{n}$ with size $p^{n-s}$ and be ordered as $S=\{w_{0}, w_{1}, \dots, w_{p^{n-s}-1}\}$. Let $d$ be a function from $V_{n-s}$ to $\mathbb{Z}_{p^k}$. For any $a\in \mathbb{F}_{p}^{n}$, define $g_{a}(x)=d(x)+p^{k-1}\psi_{a}(x)$, where $\psi_{a}$ is defined by (12). If for any $a\in \mathbb{F}_{p}^{n}$, $g_{a}: V_{n-s}\rightarrow \mathbb{Z}_{p^k}$ is a weakly regular generalized bent function and $\mu_{g_{a}}$ are the same for all $a \in \mathbb{F}_{p}^{n}$, where $\mu_{g_{a}}$ is defined by (3), then $f: \mathbb{F}_{p}^{n}\rightarrow \mathbb{Z}_{p^k}$ defined as $f(x)=g_{x}^{*}(0)$ is a generalized $s$-plateaued function with $S$ as Walsh support, where $g_{x}^{*}$ is the dual of $g_{x}$.
\end{corollary}
\begin{proof}
If for any $a\in \mathbb{F}_{p}^{n}$, $g_{a}$ is a weakly regular generalized bent function and $\mu_{g_{a}}$ are the same for all $a \in \mathbb{F}_{p}^{n}$, then $\mu_{g_{a}}(x)=u$, where $u$ is a constant independent of $a$. Define $\mu(x)=u^{-1}, x \in V_{n-s}$. It is easy to see that $\mu$ satisfy the condition of Proposition 1. For any $a \in \mathbb{F}_{p}^{n}$,
\begin{equation*}
  \sum_{x\in V_{n-s}}\mu(x)\zeta_{p^k}^{d(x)+p^{k-1}\psi_{a}(x)}=u^{-1}W_{g_{a}}(0)=u^{-1}\cdot u p^{\frac{n-s}{2}}\zeta_{p^k}^{g_{a}^{*}(0)}
  =p^{\frac{n-s}{2}}\zeta_{p^k}^{g_{a}^{*}(0)},
\end{equation*}
thus $(p^{\frac{s-n}{2}}\sum_{x\in V_{n-s}}\mu(x)\zeta_{p^k}^{d(x)+p^{k-1}\psi_{a}(x)})^{p^k}$ $=1$. By Proposition 1 and its proof, we can see that $f(x)=g_{x}^{*}(0)$ is a generalized $s$-plateaued function with $S$ as Walsh support.
\end{proof}
\subsection{Characterizing generalized plateaued functions with affine Walsh supports}
In this subsection, we characterize generalized plateaued functions whose Walsh supports are affine subspaces, which extends the case of Boolean plateaued functions \cite{Hodzic3}.
\begin{theorem}\label{Theorem1}
Keep the same notation as in Proposition 1. Let $S=t+E$ be an affine subspace of $\mathbb{F}_{p}^{n}$ and be ordered as $S=\{w_{0}, w_{1}, \dots, w_{p^{n-s}-1}\}$, where $w_{i}=t+v_{i}M , 0\leq i \leq p^{n-s}-1$, $\{v_{0}, \dots, v_{p^{n-s}-1}\}$ is the lexicographic order of $\mathbb{F}_{p}^{n-s}$ and $M$ is a matrix whose row vectors form a basis of the $(n-s)$-dimensional linear subspace $E$. Let $d$ be a function from $\mathbb{F}_{p}^{n-s}$ to $\mathbb{Z}_{p^k}$. Then the function $W$ defined by (13) is the Walsh transform of a generalized $s$-plateaued function $f: \mathbb{F}_{p}^{n}\rightarrow \mathbb{Z}_{p^k}$ if and only if $d$ is the dual of some generalized bent function $g$ and $\mu=\mu_{g}$, where $\mu_{g}$ is defined by (3). Further, if $d$ is the dual of some generalized bent function $g$ and $\mu=\mu_{g}$, then $f(x)=g(xM^{T})+p^{k-1}x\cdot t$.
\end{theorem}
\begin{proof}
For any $a\in \mathbb{F}_{p}^{n}$ and any $0\leq i \leq p^{n-s}-1$, we have $\psi_{a}(v_{i})=a \cdot w_{i}=a\cdot(t+v_{i}M)=a\cdot t+aM^{T}\cdot v_{i}$.

If $d$ is the dual of some generalized bent function $g$ and $\mu=\mu_{g}$, then we have
\begin{equation*}
  \begin{split}
     \sum_{x \in \mathbb{F}_{p}^{n-s}}\mu(x)\zeta_{p^k}^{d(x)+p^{k-1}\psi_{a}(x)} & =\sum_{x \in \mathbb{F}_{p}^{n-s}}\mu(x)\zeta_{p^k}^{d(x)}\zeta_{p}^{a \cdot t+aM^{T}\cdot x } \\
       & =\zeta_{p}^{a\cdot t}p^{\frac{n-s}{2}}\zeta_{p^k}^{g(aM^{T})},
   \end{split}
\end{equation*}
where the second equation is obtained by the inverse Walsh transform (2). Thus for any $a\in \mathbb{F}_{p}^{n}$, $(p^{\frac{s-n}{2}}\sum_{x\in \mathbb{F}_{p}^{n-s}}$ $\mu(x)\zeta_{p^k}^{d(x)+p^{k-1}\psi_{a}(x)})^{p^k}=1$. By Proposition 1 and its proof, the function $W$ defined by (13) is the Walsh transform of a generalized $s$-plateaued function $f: \mathbb{F}_{p}^{n}\rightarrow \mathbb{Z}_{p^k}$ and $f(x)=g(xM^{T})+p^{k-1}x\cdot t$.

Conversely, if the function $W$ defined by (13) is the Walsh transform of a generalized $s$-plateaued function $f: \mathbb{F}_{p}^{n}\rightarrow \mathbb{Z}_{p^k}$, by the proof of Proposition 1 we have
\begin{equation*}
  p^{\frac{s-n}{2}}\sum_{x\in \mathbb{F}_{p}^{n-s}}\mu(x)\zeta_{p^k}^{d(x)+p^{k-1}\psi_{a}(x)}=\zeta_{p^k}^{f(a)}.
\end{equation*}
Then
\begin{equation}\label{14}
  p^{\frac{s-n}{2}}\sum_{x\in \mathbb{F}_{p}^{n-s}}\mu(x)\zeta_{p^k}^{d(x)}\zeta_{p}^{aM^{T}\cdot x}=\zeta_{p^k}^{f(a)-p^{k-1}a\cdot t}.
\end{equation}
For any $y\in \mathbb{F}_{p}^{n-s}$, since $M$ is row full rank, there exists $a_{y}\in \mathbb{F}_{p}^{n}$ such that $a_{y}M^{T}=y$. When $a_{y}M^{T}=b_{y}M^{T}=y$, by (14) we have $f(a_{y})-p^{k-1}a_{y}\cdot t=f(b_{y})-p^{k-1}b_{y}\cdot t$.
Define $g:\mathbb{F}_{p}^{n-s} \rightarrow \mathbb{Z}_{p^k}$ as
\begin{equation*}
  g(y)=f(a_{y})-p^{k-1}a_{y}\cdot t,
\end{equation*}
where $a_{y} \in \mathbb{F}_{p}^{n}$ satisfies $a_{y}M^{T}=y$. Then for any $b\in \mathbb{F}_{p}^{n-s}$, by Equation (14),
\begin{equation*}
  \begin{split}
   W_{g}(b)&=\sum_{y \in \mathbb{F}_{p}^{n-s}}(p^{\frac{s-n}{2}}\sum_{x\in \mathbb{F}_{p}^{n-s}}\mu(x)\zeta_{p^k}^{d(x)}\zeta_{p}^{a_{y}M^{T}\cdot x})\zeta_{p}^{-b\cdot y}\\
           &=p^{\frac{s-n}{2}}\sum_{x\in \mathbb{F}_{p}^{n-s}}\mu(x)\zeta_{p^k}^{d(x)}\sum_{y \in \mathbb{F}_{p}^{n-s}}\zeta_{p}^{y\cdot (x-b)}\\
           &=p^{\frac{n-s}{2}}\mu(b)\zeta_{p^k}^{d(b)},\\
  \end{split}
\end{equation*}
that is, $g$ is a generalized bent function and $d$ is the dual of $g$ and $\mu_{g}=\mu$.
\end{proof}

By Theorem 1, we can see that $f: \mathbb{F}_{p}^{n}\rightarrow \mathbb{Z}_{p^k}$ is a generalized $s$-plateaued function whose Walsh support is an affine subspace if and only if there is a generalized bent function $g: \mathbb{F}_{p}^{n-s}\rightarrow \mathbb{Z}_{p^k}$, a row full rank matrix $M$ over $\mathbb{F}_{p}$ of size $(n-s) \times n$ and $t \in \mathbb{F}_{p}^{n}$ such that $f(x)=g(xM^{T})+p^{k-1}x\cdot t$. Further, if $f$ is a generalized $s$-plateaued function with affine Walsh support, then the Walsh support $S_{f}=\{w_{0}, \dots, w_{p^{n-s}-1}\}$ and the dual $f^{*}(w_{i})=d(v_{i}), 0\leq i \leq p^{n-s}-1$, where $w_{i}=t+v_{i}M$, $\mathbb{F}_{p}^{n-s}=\{v_{0}, \dots, v_{p^{n-s}-1}\}$ and $d$ is the dual of $g$. It is known that plateaued functions with affine Walsh supports correspond to partially bent functions. A function $f: V_{n}\rightarrow \mathbb{F}_{p}$ is called a partially bent function if for any $a \in V_{n}$, $f(x+a)-f(x)$ is either balanced or constant. Since plateaued functions with affine Walsh supports correspond to partially bent functions, let $k=1$, we obtain the following characterization of $p$-ary partially bent functions for any prime $p$, which extends the case of Boolean partially bent functions \cite{Hodzic3}.

\begin{corollary}\label{Corollary2}
Let $p$ be a prime, $n, s \ (< n)$ be positive integers and $n+s$ be even for $p=2$. The function $f: \mathbb{F}_{p}^{n}\rightarrow \mathbb{F}_{p}$ is a partially bent function with $|S_{f}|=p^{n-s}$ if and only if there is a bent function $g: \mathbb{F}_{p}^{n-s}\rightarrow \mathbb{F}_{p}$, a row full rank matrix $M$ over $\mathbb{F}_{p}$ of size $(n-s) \times n$ and $t \in \mathbb{F}_{p}^{n}$ such that $f(x)=g(xM^{T})+x\cdot t$. Further, if $f$ is partially bent with $|S_{f}|=p^{n-s}$, then the Walsh support $S_{f}=\{w_{0}, \dots, w_{p^{n-s}-1}\}$ and the dual  $f^{*}(w_{i})=d(v_{i}), 0\leq i \leq p^{n-s}-1$, where $w_{i}=t+v_{i}M$, $\mathbb{F}_{p}^{n-s}=\{v_{0}, \dots, v_{p^{n-s}-1}\}$ and $d$ is the dual of $g$.
\end{corollary}

We give an example of generalized plateaued function with affine Walsh support by using Theorem 1.

\begin{example}\label{Exameple1}
Let $t=(1,2,1,0,0)\in \mathbb{F}_{3}^{5}$, the row vectors of matrix $M$ over $\mathbb{F}_{3}$ be $(0, 0, 1, 1, 0), $ $(0, 1, 0, 0, 1), (1, 0, 0, 0, 2), (1, 1, 1, 1, 2)$ respectively, and $g: \mathbb{F}_{3}^{4}\rightarrow \mathbb{Z}_{3^3}$ be defined as $g(x_{1}, \dots, x_{4})$ $=3^{2}(x_{1}x_{3}+x_{2}x_{4})+3x_{3}+x_{4}$. Then $M$ is row full rank and $g$ is a Maiorana-McFarland generalized bent function. By Theorem 1, $f(x)=g(xM^{T})+3^{2}t\cdot x$ is a generalized $1$-plateaued function from $\mathbb{F}_{3}^{5}$ to $\mathbb{Z}_{3^3}$ with the Walsh support $S_{f}=\{w_{0}, \dots, w_{80}\}$, where $w_{i}=t+v_{i}M, 0\leq i \leq 80$, $\mathbb{F}_{3}^{4}=\{v_{0}, \dots, v_{80}\}$.
\end{example}
\subsection{Constructions of generalized plateaued functions with (non)-affine Walsh supports}
In this subsection, we provide some constructions of generalized plateaued functions with (non)-affine Walsh supports by spectral method.

Keep the same notation as in Proposition 1. If $f: \mathbb{F}_{p}^{n}\rightarrow \mathbb{Z}_{p^k}$ is a generalized $s$-plateaued function constructed by spectral method, by the proof of Proposition 1, we have $S_{f}=S$, where ordered $S=\{w_{0}, \dots, w_{p^{n-s}-1}\}$. It is easy to see that the matrix form of $S_{f}$ whose row vectors are $w_{0}, \dots, w_{p^{n-s}-1}$ can be written as
\begin{equation}\label{15}
  S_{f}=(T_{\psi_{a_{1}}}, \dots, T_{\psi_{a_{n}}}),
\end{equation}
where $\{a_{1}, \dots, a_{n}\}$ is the canonical basis of $\mathbb{F}_{p}^{n}$, that is, $a_{1}=(1, 0, \dots, 0, 0), \dots, $ $a_{n}=(0, 0, \dots, $ $0, 1)$, $\psi_{a_{i}}: V_{n-s}\rightarrow \mathbb{F}_{p}$ is defined by (12) and $T_{\psi_{a_{i}}}$ defined by (11) is the true table of $\psi_{a_{i}}$. If $\psi_{a_{i}}$ is an affine function, we say that the $i$-th column of (ordered) $S_{f}$ corresponds to an affine function. It is easy to see that if $f$ is constructed by Theorem 1, then every column of $S_{f}$ (ordered as in Theorem 1) corresponds to an affine function.

In \cite{Hodzic3}, Hod\v{z}i\'{c} \emph{et al}. designed Boolean plateaued functions with (non)-affine Walsh supports in spectral domain. As pointed out in \cite{Hodzic3}, for the constructions in spectral domain given in \cite{Hodzic3}, the Walsh supports of Boolean $s$-plateaued functions in $n$ variables, when written as matrices of form (15), contain at least $n-s$ columns corresponding to affine functions on $\mathbb{F}_{2}^{n-s}$. They proposed an open problem (Open Problem 2) to provide constructions of Boolean $s$-plateaued functions in $n$ variables whose Walsh supports, when written as matrices of form (15), contain strictly less than $n-s$ columns corresponding to affine functions. In the following constructions of generalized $s$-plateaued functions with non-affine Walsh supports, the Walsh supports, when written as matrices of form (15), can contain strictly less than $n-s$ columns corresponding to affine functions. When $p=2, k=1$, these constructions provide an answer to Open Problem 2 in \cite{Hodzic3}.

\begin{remark}\label{Remark1}
Let $f: \mathbb{F}_{p}^{n}\rightarrow \mathbb{Z}_{p^k}$ be a generalized $s$-plateaued function constructed by spectral method and the matrix form of the Walsh support $S_{f}$ be defined by (15). It is easy to see that if there exists $\psi_{a_{i}}$ for some $1\leq i \leq n$ which is neither balanced nor constant, or there exist $\psi_{a_{i_{1}}}, \dots, \psi_{a_{i_{n-s+1}}}$ for some $1\leq i_{1}, \dots, i_{n-s+1}\leq n$ such that for any nonzero $(b_{1}, \dots, b_{n-s+1})\in \mathbb{F}_{p}^{n-s+1}$, the function $\sum_{j=1}^{n-s+1}b_{j}\psi_{a_{i_{j}}}$ is not constant, then the Walsh support $S_{f}$ must be a non-affine subspace.
\end{remark}

In the first construction, we utilize the Maiorana-McFarland generalized bent function
\begin{equation*}
  f(x_{1}, x_{2})=p^{k-1}Tr_{1}^{n}(\alpha x_{1}\pi(x_{2}))+g(x_{2}), (x_{1}, x_{2})\in \mathbb{F}_{p^n}\times \mathbb{F}_{p^n},
\end{equation*}
where $\alpha \in \mathbb{F}_{p^n}^{*}$, $\pi$ is a permutation over $\mathbb{F}_{p^n}$ and $g$ is an arbitrary function from $\mathbb{F}_{p^n}$ to $\mathbb{Z}_{p^k}$. By (6), $f^{*}(x_{1}, x_{2})=p^{k-1}Tr_{1}^{n}(-\pi^{-1}(\alpha^{-1}x_{1})x_{2})+g(\pi^{-1}(\alpha^{-1}x_{1}))$ and $\mu_{f}(x_{1}, x_{2})=1$, where $f^{*}$ is the dual of $f$ and $\mu_{f}$ is defined by (3).

For the sake of simplicity, we give the functions needed in the following theorem. Let $n, s \ ( < n)$ be positive integers with $n-s=2m$, $\{\alpha_{1}, \dots, \alpha_{m}\}$ be a basis of $\mathbb{F}_{p^m}$ over $\mathbb{F}_{p}$, $\pi$ be a permutation over $\mathbb{F}_{p^m}$ and $L_{1}, \dots, L_{n-s}: \mathbb{F}_{p^m}\times \mathbb{F}_{p^m}\rightarrow \mathbb{F}_{p}$ be linearly independent linear functions. Define $d: \mathbb{F}_{p^m}\times \mathbb{F}_{p^m}\rightarrow \mathbb{Z}_{p^k}$ as
\begin{equation}\label{16}
  d(x_{1}, x_{2})=p^{k-1}Tr_{1}^{m}(\alpha_{1}x_{1}\pi(x_{2}))+g(x_{2}),
\end{equation}
where $g$ is an arbitrary function from $\mathbb{F}_{p^m}$ to $\mathbb{Z}_{p^k}$. Define $t_{i}: \mathbb{F}_{p^m}\times \mathbb{F}_{p^m}\rightarrow \mathbb{F}_{p}, 1\leq i \leq s$ as
\begin{equation}\label{17}
  t_{i}(x_{1}, x_{2})=\left\{\begin{split}
                                Tr_{1}^{m}(\beta_{i}x_{1}\pi(x_{2}))+g_{i}(x_{2})+A_{i}(x_{1}, x_{2}) & \ \ \text{if} \ m \geq 2, \\
                               g_{i}(x_{2})+A_{i}(x_{1}, x_{2}) & \ \ \text{if} \ m=1,
                             \end{split}\right.
\end{equation}
where $\beta_{i}=\sum_{j=2}^{m}c_{i, j}\alpha_{j}$ with $c_{i, j}\in \mathbb{F}_{p}$, $g_{i}$ is an arbitrary function from $\mathbb{F}_{p^m}$ to $\mathbb{F}_{p}$ and $A_{i}$ is an arbitrary affine function from $\mathbb{F}_{p^m} \times \mathbb{F}_{p^m}$ to $\mathbb{F}_{p}$. Define $h_{j}: \mathbb{F}_{p^m}\times \mathbb{F}_{p^m}\rightarrow \mathbb{F}_{p}, 1\leq j \leq n-s$ as
\begin{equation}\label{18}
  h_{j}=\left\{\begin{split}
                               \sum_{i=1}^{s}d_{j, i}t_{i}+L_{j}+b_{j} & \ \ \text{if} \ I=\emptyset, \\
                               \sum_{i \notin I}d_{j, i}t_{i}+F_{j}(t_{i_{1}}, \dots, t_{i_{|I|}})+L_{j}+b_{j}& \ \ \text{if} \ I \neq \emptyset,
                             \end{split}\right.
\end{equation}
where $I=\{1\leq i \leq s: t_{i}(x_{1}, x_{2}) \ \text{only  depends  on  variable} \ x_{2}\}$ and denote $I$ by $\{i_{1}, \dots, i_{|I|}\}$ if $I\neq \emptyset$, $d_{j, i}, b_{j} \in \mathbb{F}_{p}$ and $F_{j}$ is an arbitrary function from $\mathbb{F}_{p}^{|I|}$ to $\mathbb{F}_{p}$.
\begin{theorem}\label{Theorem2}
Let $p$ be a prime, $n, s \ ( < n)$ be positive integers with $n-s=2m$. Let $d: \mathbb{F}_{p^m}\times \mathbb{F}_{p^m}\rightarrow \mathbb{Z}_{p^k}$ be defined by (16). Let the matrix form of $S=\{w_{0}, \dots, w_{p^{n-s}-1}\}\subseteq \mathbb{F}_{p}^{n}$ be defined as
\begin{equation*}
S=\left(\begin{array}{c}
            w_{0} \\
            \dots \\
            w_{p^{n-s}-1} \\
          \end{array}
        \right)
=(T_{t_{1}}, \dots, T_{t_{s}}, T_{h_{1}}, \dots, T_{h_{n-s}}),
\end{equation*}
where $t_{i} \ (1\leq i \leq s)$ are defined by (17) and $h_{j} \ (1\leq j \leq n-s)$ are defined by (18).
Then $f(x)=(d(y)+p^{k-1}\psi_{x}(y))^{*}(0)$ is a generalized $s$-plateaued function from $\mathbb{F}_{p}^{n}$ to $\mathbb{Z}_{p^k}$ with $S$ as Walsh support, where $\psi_{x}$ is defined by (12), $(d(y)+p^{k-1}\psi_{x}(y))^{*}$ is the dual of generalized bent function $d(y)+p^{k-1}\psi_{x}(y)$.
\end{theorem}
\begin{proof}
First we show that the size of $S$ is equal to $p^{n-s}$, that is to prove
\begin{equation*}
  (t_{1}(x), \dots, h_{n-s}(x))=(t_{1}(x'), \dots, h_{n-s}(x')) \iff \ x=x',
\end{equation*}
where $x=(x_{1}, x_{2}), x'=(x'_{1}, x'_{2})\in \mathbb{F}_{p^m}\times \mathbb{F}_{p^m}$. If $(t_{1}(x), \dots, h_{n-s}(x))=(t_{1}(x'), \dots, h_{n-s}(x'))$, then by the definitions of $h_{j} \ (1\leq j \leq n-s)$, $L_{j}(x)=L_{j}(x')$ for any $1\leq j \leq n-s$. Since $L_{1}, \dots, L_{n-s}$ are linearly independent linear functions, we can see that $x=x'$.

For any $a\in \mathbb{F}_{p}^{n}$ and $0\leq i \leq p^{n-s}-1$, $\psi_{a}(v_{i})=a\cdot w_{i}=a\cdot (t_{1}(v_{i}), \dots, t_{s}(v_{i}), h_{1}(v_{i}), \dots, $ $h_{n-s}(v_{i}))$. When $m\geq 2$, by the constructions of $t_{i}, h_{j} \ (1\leq i \leq s, 1\leq j \leq n-s)$, we have $\psi_{a}(x_{1}, x_{2})=Tr_{1}^{m}(\alpha_{a}x_{1}\pi(x_{2}))+g_{a}(x_{2})+A_{a}(x_{1}, x_{2})$, where $\alpha_{a}\in \mathbb{F}_{p^m}$ is some linear combination of $\alpha_{2}, \dots, \alpha_{m}$, $g_{a}$ is some function from $\mathbb{F}_{p^m}$ to $\mathbb{F}_{p}$ and $A_{a}: \mathbb{F}_{p^m}\times \mathbb{F}_{p^m}\rightarrow \mathbb{F}_{p}$ is some affine function. Then $d(x_{1}, x_{2})+p^{k-1}\psi_{a}(x_{1}, x_{2})=p^{k-1}Tr_{1}^{m}((\alpha_{1}+\alpha_{a})x_{1}\pi(x_{2}))+(g(x_{2})+p^{k-1}g_{a}(x_{2}))+p^{k-1}A_{a}(x_{1}, x_{2})$. Since $\alpha_{1}, \dots, \alpha_{m}$ are linearly independent, $\alpha_{1}+\alpha_{a}\neq 0$. It is easy to see that if $h: V_{n}\rightarrow \mathbb{Z}_{p^k}$ is a weakly regular generalized bent function and $A: V_{n}\rightarrow \mathbb{F}_{p}$ is an arbitrary affine function, then $h+p^{k-1}A$ is also a weakly regular generalized bent function and $\mu_{h+p^{k-1}A}=\mu_{h}$. Hence, $d+p^{k-1}\psi_{a}$ is a weakly regular generalized bent function and $\mu_{d+p^{k-1}\psi_{a}}=1$ for any $a \in \mathbb{F}_{p}^{n}$. By Corollary 1, $f(x)=(d+p^{k-1}\psi_{x})^{*}(0)$ is a generalized $s$-plateaued function from $\mathbb{F}_{p}^{n}$ to $\mathbb{Z}_{p^k}$ with $S$ as Walsh support, where $\psi_{x}$ is defined by (12). When $m=1$, by the similar argument, we have the same conclusion.
\end{proof}

\begin{remark}\label{Remark2}
Note that by (6), computing $(d+p^{k-1}\psi_{x})^{*}(0)$ is routine. When $k=1$, Theorem 2 is applicable for constructing $p$-ary plateaued functions for any prime $p$. Theorem 2 extends the construction of Theorem 4.1 of \cite{Hodzic3} for Boolean plateaued functions and can be used to construct generalized plateaued functions for which the matrix form of the Walsh support defined by (15) can contain strictly less than $n-s$ columns corresponding to affine functions, which provides an answer to Open Problem 2 in \cite{Hodzic3} when $p=2, k=1$.
\end{remark}

We give two examples by using Theorem 2. The first example gives a generalized $3$-ary plateaued function and the second example gives a Boolean plateaued function. Both of them satisfy that the Walsh support is non-affine and every column of the matrix form of the Walsh support defined by (15) corresponds to a non-affine function. Furthermore, the constructed Boolean plateaued function has no nonzero linear structure. For a Boolean function $f: V_{n} \rightarrow \mathbb{F}_{2}$, if $f(x)+f(x+a)$ is a constant function, then $a$ is called a linear structure of $f$.
\begin{example}\label{Example2}
Let $p=3, k=2, n=7, s=3$. Let $z$ be the primitive element of $\mathbb{F}_{3^2}$ with $z^2+2z+2=0$. Let $d: \mathbb{F}_{3^2} \times \mathbb{F}_{3^2}\rightarrow \mathbb{Z}_{3^2}$ be defined by $d(y_{1}, y_{2})=3Tr_{1}^{2}(zy_{1}y_{2})+2(Tr_{1}^{2}(y_{2}))^{2}$. Let $t_{1}(y_{1}, y_{2})=Tr_{1}^{2}(y_{1}y_{2})$, $t_{2}(y_{1}, y_{2})=Tr_{1}^{2}(y_{2}^{2})$, $t_{3}(y_{1}, y_{2})=Tr_{1}^{2}(zy_{2}^{2})$, $h_{1}=t_{1}+Tr_{1}^{2}(y_{1})$, $h_{2}=t_{2}^{2}+Tr_{1}^{2}(zy_{1})$, $h_{3}=t_{3}^{2}+Tr_{1}^{2}(y_{2})$, $h_{4}=t_{2}+t_{3}+Tr_{1}^{2}(zy_{2})$, $(y_{1}, y_{2})\in \mathbb{F}_{3^2}\times \mathbb{F}_{3^2}$. Then by computing $(d+3\psi_{x})^{*}(0)$, we can obtain generalized $3$-plateaued function $f(x_{1}, \dots, x_{7})=2(((x_{1}+x_{4})^{2}x_{5}+(2(x_{1}+x_{4})+1)(x_{4}+x_{5})) mod \ 3)^{2}+3((x_{1}+x_{4})^{2}((x_{2}+x_{7})(2x_{4}^{2}+2x_{4}x_{5})+(x_{3}+x_{7})(x_{4}^{2}+x_{5}^{2})+2x_{4}^{2}x_{5}+2x_{4}x_{5}^{2}+x_{4}x_{7}+x_{5}x_{6}
+x_{5}x_{7}+x_{5})+(x_{1}+x_{4})((x_{2}+x_{7})(2x_{4}^{2}+x_{4}x_{5}+x_{5}^{2})+(x_{3}+x_{7})(2x_{4}^{2}+2x_{4}x_{5})+2x_{4}^{2}x_{5}+2x_{4}x_{5}^{2}+2x_{4}x_{6}
+2x_{4}x_{7}+2x_{5}x_{6}+x_{5}x_{7}+x_{5})+2x_{4}^{2}x_{5}^{2}x_{6}+(x_{2}+x_{7})(x_{4}x_{5}+2x_{5}^{2})+(x_{3}+x_{7})(2x_{4}^{2}+x_{4}x_{5}+x_{5}^{2})+x_{4}^{2}x_{5}+x_{4}^{2}x_{6}
+x_{4}x_{5}^{2}+x_{5}^{2}x_{6}+x_{4}x_{6}+x_{4}x_{7}+x_{5}x_{6}+2x_{5}x_{7}+x_{5})$ from $\mathbb{F}_{3}^{7}$ to $\mathbb{Z}_{3^2}$. Since $t_{1}$ is neither balanced nor constant, the Walsh support $S_{f}$ is not an affine subspace. Since $t_{i}\ (1\leq i \leq 3), h_{j}\ (1\leq j \leq 4)$ are all non-affine functions and the matrix form of $S_{f}$ defined by (15) is $S_{f}=(T_{t_{1}}, \dots, T_{t_{3}}, T_{h_{1}}, \dots, T_{h_{4}})$, every column of $S_{f}$ corresponds to a non-affine function.
\end{example}

\begin{example}\label{Example3}
Let $p=2, k=1, n=10, s=4$. Let $z$ be the primitive element of $\mathbb{F}_{2^3}$ with $z^{3}+z+1=0$. Let $d(y_{1}, y_{2})=Tr_{1}^{3}(z^{2}y_{1}y_{2})$, $t_{1}(y_{1}, y_{2})=Tr_{1}^{3}(y_{1}y_{2})$, $t_{2}(y_{1}, y_{2})=Tr_{1}^{3}(zy_{1}y_{2})$, $t_{3}(y_{1}, y_{2})=Tr_{1}^{3}(y_{2}^{3})$, $t_{4}(y_{1}, y_{2})=Tr_{1}^{3}(zy_{2}^{3})$, $h_{1}=t_{1}+Tr_{1}^{3}(y_{1})$, $h_{2}=t_{1}+Tr_{1}^{3}(zy_{1})$, $h_{3}=t_{2}+Tr_{1}^{3}(z^{2}y_{1})$, $h_{4}=t_{2}+Tr_{1}^{3}(y_{2})$, $h_{5}=t_{3}t_{4}+Tr_{1}^{3}(zy_{2})$, $h_{6}=t_{3}t_{4}+Tr_{1}^{3}(z^{2}y_{2})$, $(y_{1}, y_{2})\in \mathbb{F}_{2^3}\times \mathbb{F}_{2^3}$. Then by computing $(d+\psi_{x})^{*}(0)$, we can obtain Boolean $4$-plateaued function $f(x_{1}, \dots, x_{10})=(x_{1}+x_{5}+x_{6}+1)(x_{3}(x_{5}x_{7}+x_{6}x_{7}+x_{5})+x_{4}(x_{5}x_{6}+x_{5}x_{7}+x_{6}x_{7}+x_{5}+x_{7})+(x_{5}x_{6}+x_{5}x_{7})(x_{9}+x_{10})
+x_{5}x_{8}+x_{6}x_{10}+x_{7}x_{9})+((x_{1}+x_{5}+x_{6})(x_{2}+x_{7}+x_{8})+1)(x_{5}x_{9}+x_{6}x_{9}+x_{7}x_{8}+x_{7}x_{9})+(x_{1}+x_{2}+x_{5}+x_{6}+x_{7}+x_{8}+1)
(x_{3}(x_{5}x_{7}+x_{6}+x_{7})+x_{4}(x_{5}x_{7}+x_{6}x_{7}+x_{5})+x_{5}x_{6}(x_{9}+x_{10})+x_{5}x_{9}+x_{6}x_{8}+x_{6}x_{9}+x_{7}x_{9}+x_{7}x_{10})
+x_{3}(x_{5}x_{6}+x_{6}x_{7}+x_{5}+x_{6})+x_{4}(x_{5}x_{7}+x_{6}+x_{7})+(x_{6}x_{7}+x_{5}+x_{6}+x_{7})(x_{9}+x_{10})$. Since $t_{1}$ is neither balanced nor constant, the Walsh support $S_{f}$ is not an affine subspace. Since $t_{i}\ (1\leq i \leq 4), h_{j}\ (1\leq j \leq 6)$ are all non-affine functions and the matrix form of $S_{f}$ defined by (15) is $S_{f}=(T_{t_{1}}, \dots, T_{t_{4}}, T_{h_{1}}, \dots, T_{h_{6}})$, every column of $S_{f}$ corresponds to a non-affine function. Furthermore, one can verify that $S_{f}$ contains a basis of $\mathbb{F}_{2}^{10}$ and $(0, \dots, 0)\in S_{f}$, hence by Corollary 3.1 of \cite{Hodzic3}, $f$ has no nonzero linear structure.
\end{example}

In the second construction, we take advantage of the good properties of general generalized bent functions given in \cite{Mesnager9}. Let $t\geq 2$ be an integer. Let $f(x)=\sum_{i=0}^{t-1}p^{t-1-i}f_{i}(x)$ with $f_{i}: V_{n}\rightarrow \mathbb{F}_{p}, 0\leq i \leq t-1$ be a generalized bent function from $V_{n}$ to $\mathbb{Z}_{p^t}$, where $p$ is an odd prime or $p=2$ and $n$ is even. Then by Corollary 7 of \cite{Mesnager9}, for any function $G: \mathbb{F}_{p}^{t-1}\rightarrow \mathbb{Z}_{p^{k}}$, the function $p^{k-1}f_{0}+G(f_{1}, \dots, f_{t-1})$ is a generalized bent function from $V_{n}$ to $\mathbb{Z}_{p^k}$  with $\mu_{p^{k-1}f_{0}+G(f_{1}, \dots, f_{t-1})}=\mu_{f}$, where $\mu_{f}$ is defined by (3).

For the sake of simplicity, we give the functions needed in the following theorem. Let $n, s \ (< n)$ be positive integers with $n-s$ even if $p=2$, $L_{1}, \dots, L_{n-s}: V_{n-s}\rightarrow \mathbb{F}_{p}$ be linearly independent linear functions and $g=\sum_{i=0}^{t-1}p^{t-1-i}g_{i}$ with $g_{i}: V_{n-s}\rightarrow \mathbb{F}_{p}, 0\leq i \leq t-1$ be a weakly regular generalized bent function from $V_{n-s}$ to $\mathbb{Z}_{p^t}$, where $t \geq 2$. Define $d: V_{n-s}\rightarrow \mathbb{Z}_{p^k}$ as
\begin{equation}\label{19}
  d(x)=p^{k-1}g_{0}(x)+G(g_{1}(x), \dots, g_{t-1}(x)),
\end{equation}
where $G$ is an arbitrary function from $\mathbb{F}_{p}^{t-1}$ to $\mathbb{Z}_{p^k}$.
Define $t_{i}: V_{n-s}\rightarrow \mathbb{F}_{p}, 1\leq i \leq s$ as
\begin{equation}\label{20}
  t_{i}(x)=F_{i}(g_{1}(x), \dots, g_{t-1}(x)),
\end{equation}
where $F_{i}$ is an arbitrary function from $\mathbb{F}_{p}^{t-1}$ to $\mathbb{F}_{p}$. Define $h_{j}: V_{n-s}\rightarrow \mathbb{F}_{p}, 1\leq j \leq n-s$ as
\begin{equation}\label{21}
  h_{j}(x)=H_{j}(t_{1}(x), \dots, t_{s}(x))+L_{j}(x)+b_{j},
\end{equation}
where $H_{j}$ is an arbitrary function from $\mathbb{F}_{p}^{s}$ to $\mathbb{F}_{p}$ and $b_{j} \in \mathbb{F}_{p}$.
\begin{theorem}\label{Theorem3}
Let $p$ be a prime, $n, s \ (< n)$ be positive integers with $n-s$ even if $p=2$. Let $d: V_{n-s} \rightarrow \mathbb{Z}_{p^k}$ be defined by (19). Let the matrix form of $S=\{w_{0}, \dots,$ $w_{p^{n-s}-1}\}\subseteq \mathbb{F}_{p}^{n}$ be defined as
\begin{equation*}
S=\left(
          \begin{array}{c}
            w_{0} \\
            \dots \\
            w_{p^{n-s}-1} \\
          \end{array}
        \right)
=(T_{t_{1}}, \dots, T_{t_{s}}, T_{h_{1}}, \dots, T_{h_{n-s}}),
\end{equation*}
where $t_{i} \ (1\leq i \leq s)$ are defined by (20) and $h_{j} \ (1\leq j \leq n-s)$ are defined by (21).
Then $f(x)=(d(y)+p^{k-1}\psi_{x}(y))^{*}(0)$ is a generalized $s$-plateaued function from $\mathbb{F}_{p}^{n}$ to $\mathbb{Z}_{p^k}$ with $S$ as Walsh support, where $\psi_{x}$ is defined by (12), $(d(y)+p^{k-1}\psi_{x}(y))^{*}$ is the dual of generalized bent function $d(y)+p^{k-1}\psi_{x}(y)$.
\end{theorem}
\begin{proof}
With the similar argument as in the proof of Theorem 2, we have $|S|=p^{n-s}$ and for any $a \in \mathbb{F}_{p}^{n}$, $\psi_{a}(x)=G_{a}(g_{1}(x), \dots, g_{t-1}(x))+A_{a}(x)$, where $G_{a}$ is some function from $\mathbb{F}_{p}^{t-1}$ to $\mathbb{F}_{p}$ and $A_{a}: V_{n-s}\rightarrow \mathbb{F}_{p}$ is some affine function. Then $d+p^{k-1}\psi_{a}$ is a weakly regular generalized bent function and $\mu_{d+p^{k-1}\psi_{a}}=\mu_{g}$ for any $a \in \mathbb{F}_{p}^{n}$. By Corollary 1, $f(x)=(d+p^{k-1}\psi_{x})^{*}(0)$ is a generalized $s$-plateaued function from $\mathbb{F}_{p}^{n}$ to $\mathbb{Z}_{p^k}$ with $S$ as Walsh support, where $\psi_{x}$ is defined by (12).
\end{proof}

\begin{remark}\label{Remark3}
When $k=1$, Theorem 3 is applicable for constructing $p$-ary plateaued functions for any prime $p$. Theorem 3 can be used to construct generalized plateaued functions for which the matrix form of the Walsh support defined by (15) can contain strictly less than $n-s$ columns corresponding to affine functions, which provides an answer to Open Problem 2 in \cite{Hodzic3} when $p=2, k=1$.
\end{remark}

We give two examples by using Theorem 3. The first example gives a generalized $5$-ary plateaued function and the second example gives a Boolean plateaued function. Both of them satisfy that the Walsh support is non-affine and every column of the matrix form of the Walsh support defined by (15) corresponds to a non-affine function. Furthermore, the constructed Boolean plateaued function has no nonzero linear structure.
\begin{example}\label{Example4}
Let $p=5, k=3, n=4, s=1, t=2$. Let $z$ be the primitive element of $\mathbb{F}_{5^3}$ with $z^3+3z+3=0$. Let $g: \mathbb{F}_{5^3}\rightarrow \mathbb{Z}_{5^2}$ be defined by $g=5g_{0}+g_{1}$, $g_{0}, g_{1}: \mathbb{F}_{5^3}\rightarrow \mathbb{F}_{5}$, where $g_{0}(y)=Tr_{1}^{3}(2y^2)$, $g_{1}(y)=Tr_{1}^{3}(z^{16} y)$. Then by Theorem 16 of \cite{Mesnager10} and Corollary 3 of \cite{Qi}, $g$ is a weakly regular generalized bent function. Let $d: \mathbb{F}_{5^3}\rightarrow \mathbb{Z}_{5^3}$ be defined by $d(y)=25g_{0}(y)+g_{1}^{4}(y)$. Let $t_{1}(y)=g_{1}^{3}(y)$, $h_{1}(y)=t_{1}^{2}(y)+Tr_{1}^{3}(y)$, $h_{2}(y)=t_{1}^{4}(y)+Tr_{1}^{3}(zy)$, $h_{3}(y)=t_{1}(y)+Tr_{1}^{3}(z^2 y)$, $y \in \mathbb{F}_{5^3}$. Then by computing $(d+25\psi_{x})^{*}(0)$, we can obtain generalized $1$-plateaued function $f(x_{1}, \dots, x_{4})=((x_{2}-x_{4}) \ mod \ 5)^{4}+25(x_{3}(x_{2}-x_{4})^{4}+(x_{1}+x_{4})(x_{2}-x_{4})^{3}+x_{2}(x_{2}-x_{4})^{2}-x_{2}^{2}-x_{2}x_{4}+2x_{3}^{2}+x_{3}x_{4}-x_{4}^{2})$ from $\mathbb{F}_{5}^{4}$ to $\mathbb{Z}_{5^3}$. Since $t_{1}(0)=h_{j}(0)=0, 1\leq j \leq 3$ and $t_{1}, h_{1}, h_{2}, h_{3}$ are linearly independent, the Walsh support $S_{f}$ is non-affine. Since $t_{1}, h_{j}\ (1\leq j \leq 3)$ are all non-affine functions and the matrix form of $S_{f}$ defined by (15) is $S_{f}=(T_{t_{1}}, T_{h_{1}}, \dots, T_{h_{3}})$, every column of $S_{f}$ corresponds to a non-affine function.
\end{example}
\begin{example}\label{Example5}
Let $p=2, k=1, n=8, s=2, t=3$. Let $g: \mathbb{F}_{2}^{6}\rightarrow \mathbb{Z}_{2^3}$ be defined by $g=\sum_{i=0}^{2}2^{2-i}g_{i}, g_{i}: \mathbb{F}_{2}^{6}\rightarrow \mathbb{F}_{2}$, where $g_{0}(y_{1}, \dots, y_{6})=y_{1}y_{3}+y_{2}y_{4}+y_{5}y_{6}$, $g_{1}(y_{1}, \dots, y_{6})=y_{1}y_{2}y_{6}+y_{3}(y_{6}+1)$, $g_{2}(y_{1}, \dots, y_{6})=y_{3}y_{4}(y_{6}+1)+y_{1}y_{6}$. Then $g$ is a generalized Boolean bent function (which will be constructed in Example 14). Let $d=g_{0}$, $t_{1}=g_{1}$, $t_{2}=g_{2}$, $h_{1}=t_{1}t_{2}+y_{1}$, $h_{2}=t_{1}+y_{2}$, $h_{3}=t_{1}+y_{3}$, $h_{4}=t_{2}+y_{4}$, $h_{5}=t_{2}+y_{5}$, $h_{6}=t_{1}+t_{2}+y_{6}$. Then by computing $(d+\psi_{x})^{*}(0)$, we can obtain Boolean $2$-plateaued function
$f(x_{1}, \dots, x_{8})=x_{3}x_{5}+x_{4}x_{6}+x_{7}x_{8}+x_{3}(x_{7}+1)(x_{2}x_{4}+x_{4}x_{6}+x_{4}x_{8}+x_{1}+x_{5}+x_{8})+x_{5}x_{7}(x_{1}x_{6}+x_{3}x_{6}+x_{4}x_{6}+x_{6}x_{8}+x_{2}+x_{8}+1)$.
Since $t_{1}$ is neither balanced nor constant, the Walsh support $S_{f}$ is not an affine subspace. Since $t_{i}\ (1\leq i \leq 2), h_{j}\ (1 \leq j \leq 6)$ are all non-affine functions and the matrix form of $S_{f}$ defined by (15) is $S_{f}=(T_{t_{1}}, T_{t_{2}}, T_{h_{1}}, \dots, T_{h_{6}})$, every column of $S_{f}$ corresponds to a non-affine function. Furthermore, one can verify that $S_{f}$ contains a basis of $\mathbb{F}_{2}^{8}$ and $(0, \dots, 0)\in S_{f}$, hence by Corollary 3.1 of \cite{Hodzic3}, $f$ has no nonzero linear structure.
\end{example}

The third construction is used to construct plateaued functions, that is, $k=1$. In the following theorem, we utilize vectorial bent functions. A function $f=(f_{1}, \dots, f_{m}): V_{n} \rightarrow \mathbb{F}_{p}^{m}$ is called a vectorial bent function if for any nonzero vector $(a_{1}, \dots, a_{m}) \in \mathbb{F}_{p}^{m}$, $\sum_{i=1}^{m}a_{i}f_{i}(x): V_{n}\rightarrow \mathbb{F}_{p}$ is a bent function. It is known that if $f=(f_{1}, \dots, f_{m}): V_{n} \rightarrow \mathbb{F}_{p}^{m}$ is vectorial bent, then $m\leq n$ if $p$ is an odd prime, and $n$ is even and $m\leq \frac{n}{2}$ if $p=2$. The following theorem generalizes Theorem 4.3 of \cite{Hodzic3} for Boolean plateaued functions and can be applied to construct $s$-plateaued functions in $n$ variables whose Walsh supports, when written as matrices of form (15), contain strictly less than $n-s$ columns corresponding to affine functions. Thus, when $p=2$, the following theorem provide an answer to Open Problem 2 in \cite{Hodzic3}.

For the sake of simplicity, we give the functions needed in the following theorem. Let $n, s \ (< n), m$ be positive integers with $2\leq m \leq n-s$ if $p$ is an odd prime, and $n-s$ even and $2\leq m\leq \frac{n-s}{2}$ if $p=2$. Let $g=(g_{1}, \dots, g_{m})$ be a vectorial bent function from $V_{n-s}$ to $\mathbb{F}_{p}^{m}$ which satisfies that for any $(c_{2}, \dots, c_{m})\in \mathbb{F}_{p}^{m-1}$, $\mu_{g_{1}+\sum_{i=2}^{m}c_{i}g_{i}}(x)=u, x \in V_{n-s}$, where $\mu_{g_{1}+\sum_{i=2}^{m}c_{i}g_{i}}$ is defined by (3) and $u$ is a constant independent of $(c_{2}, \dots, c_{m})$. Let $L_{1}, \dots, L_{n-s}: V_{n-s}\rightarrow \mathbb{F}_{p}$ be linearly independent linear functions. Define $d: V_{n-s}\rightarrow \mathbb{F}_{p}$ as
\begin{equation}\label{22}
  d(x)=g_{1}(x).
\end{equation}
Define $t_{i}: V_{n-s}\rightarrow \mathbb{F}_{p}, 1\leq i \leq s$ as
\begin{equation}\label{23}
  t_{i}(x)=\sum_{j=2}^{m}c_{i, j}g_{j}(x)+A_{i}(x),
\end{equation}
where $c_{i, j} \in \mathbb{F}_{p}$, $A_{i}$ is an arbitrary affine function from $V_{n-s}$ to $\mathbb{F}_{p}$. Define $h_{j}: V_{n-s}\rightarrow \mathbb{F}_{p}, 1\leq j\leq n-s$ as
\begin{equation}\label{24}
  h_{j}(x)=\sum_{i=1}^{s}d_{j, i}t_{i}(x)+L_{j}(x)+b_{j},
\end{equation}
where $d_{j, i}, b_{j} \in \mathbb{F}_{p}$.
\begin{theorem}\label{Theorem4}
Let $p$ be a prime, $n, s \ (< n), m$ be positive integers with $2\leq m \leq n-s$ if $p$ is an odd prime, and $n-s$ even and $2\leq m\leq \frac{n-s}{2}$ if $p=2$. Let $d: V_{n-s} \rightarrow \mathbb{F}_{p}$ be defined by (22). Let the matrix form of $S=\{w_{0}, \dots,$ $w_{p^{n-s}-1}\}\subseteq \mathbb{F}_{p}^{n}$ be defined as
\begin{equation*}
S=\left(
          \begin{array}{c}
            w_{0} \\
            \dots \\
            w_{p^{n-s}-1} \\
          \end{array}
        \right)
=(T_{t_{1}}, \dots, T_{t_{s}}, T_{h_{1}}, \dots, T_{h_{n-s}}),
\end{equation*}
where $t_{i} \ (1\leq i \leq s)$ are defined by (23) and $h_{j} \ (1\leq j \leq n-s)$ are defined by (24). Then $f(x)=(d(y)+\psi_{x}(y))^{*}(0)$ is an $s$-plateaued function from $\mathbb{F}_{p}^{n}$ to $\mathbb{F}_{p}$ with $S$ as Walsh support, where $\psi_{x}$ is defined by (12), $(d(y)+\psi_{x}(y))^{*}$ is the dual of bent function $d(y)+\psi_{x}(y)$.
\end{theorem}
\begin{proof}
With the similar argument as in Theorem 2, we have $|S|=p^{n-s}$ and for any $a \in \mathbb{F}_{p}^{n}$, $\psi_{a}(x)=L_{a}(g_{2}(x), \dots, g_{m}(x))+A_{a}(x)$, where $L_{a}$ is some linear function from $\mathbb{F}_{p}^{m-1}$ to $\mathbb{F}_{p}$ and $A_{a}: V_{n-s}\rightarrow \mathbb{F}_{p}$ is some affine function. Then $d+\psi_{a}$ is a weakly regular bent function and $\mu_{d+\psi_{a}}=u$ for any $a\in \mathbb{F}_{p}^{n}$. By Corollary 1, $f(x)=(d(y)+\psi_{x}(y))^{*}(0)$ is an $s$-plateaued function with $S$ as Walsh support, where $\psi_{x}$ is defined by (12).
\end{proof}

\section{Generalized indirect sum for constructing generalized plateaued functions}
In \cite{Carlet3}, Carlet provided the so-called indirect sum for constructing Boolean bent functions, which is also applicable for constructing Boolean plateaued functions. The indirect sum use arbitrary two Boolean plateaued functions and arbitrary two Boolean bent functions as building blocks. As far as we know, up to now, there is no $p$-ary version of the indirect sum, where $p$ is an odd prime. In this section, we consider to present a $p$-ary version of the indirect sum for constructing generalized plateaued functions (we call it generalized indirect sum), which extends the indirect sum \cite{Carlet3}. Differently to Boolean case, for $p$-ary case, the bent functions used in the generalized indirect sum need to satisfy some extra conditions. We will illustrate that although the corresponding conditions for used bent functions seem harsh, we can still provide abundant bent functions that satisfy such conditions.
\begin{theorem}\label{Theorem5}
Let $p$ be a prime, $k, t, r, m$ be positive integers, $s \ ( \leq r)$ be a non-negative integer and $m$ be even for $p=2$, $r+s$ be even for $p=2, k=1$. Let $f_{c} \ (c \in \mathbb{F}_{p}^{t}): V_{r} \rightarrow \mathbb{Z}_{p^k}$ be generalized $s$-plateaued functions. Let $g_{i} \ (0\leq i \leq t): V_{m} \rightarrow \mathbb{F}_{p}$ be bent functions which satisfy that for any $j=(j_{1}, \dots, j_{t})\in \mathbb{F}_{p}^{t}$,
$$G_{j}\triangleq (1-j_{1}-\dots-j_{t})g_{0}+j_{1}g_{1}+\dots+j_{t}g_{t}$$ is a bent function and $$G_{j}^{*}=(1-j_{1}-\dots-j_{t})g_{0}^{*}+j_{1}g_{1}^{*}+\dots+j_{t}g_{t}^{*}$$ and $\mu_{G_{j}}=u$, where $\mu_{G_{j}}$ is defined by (3) and $u$ is a function from $V_{m}$ to $\{\pm 1, \pm \sqrt{-1}\}$ independent of $j$. Let $g: \mathbb{F}_{p}^{t} \rightarrow \mathbb{Z}_{p^k}$ be an arbitrary function. Then $F: V_{r}\times V_{m} \rightarrow \mathbb{Z}_{p^k}$ defined as $F(x,y)=f_{(g_{0}(y)-g_{1}(y), \dots, g_{0}(y)-g_{t}(y))}(x)+p^{k-1}g_{0}(y)+g(g_{0}(y)-g_{1}(y), \dots, g_{0}(y)-g_{t}(y))$ is a generalized $s$-plateaued function.
\end{theorem}
\begin{proof}
For any $(a, b) \in V_{r}\times V_{m}$, we have
\begin{equation}
  \begin{split}
    &W_{F}(a, b)\\
    &=\sum_{x \in V_{r}, y \in V_{m}}\zeta_{p^k}^{f_{(g_{0}(y)-g_{1}(y), \dots, g_{0}(y)-g_{t}(y))}(x)+p^{k-1}g_{0}(y)+g(g_{0}(y)-g_{1}(y), \dots, g_{0}(y)-g_{t}(y))}\zeta_{p}^{-\langle a, x\rangle -\langle b, y \rangle}\\
    & =\sum_{i_{1}, \dots, i_{t}\in \mathbb{F}_{p}}\sum_{y: g_{0}(y)-g_{j}(y)=i_{j}, 1\leq j\leq t}\sum_{x \in V_{r}}\zeta_{p^k}^{f_{(i_{1}, \dots, i_{t})}(x)+g(i_{1}, \dots, i_{t})}\zeta_{p}^{g_{0}(y)-\langle a, x \rangle-\langle b, y\rangle}\\
    & =p^{-t}\sum_{i_{1}, \dots, i_{t} \in \mathbb{F}_{p}}\zeta_{p^k}^{g(i_{1}, \dots, i_{t})}W_{f_{(i_{1}, \dots, i_{t})}}(a)\sum_{y \in V_{m}}\zeta_{p}^{g_{0}(y)-\langle b, y\rangle}\sum_{j_{1}\in \mathbb{F}_{p}}\zeta_{p}^{(i_{1}-(g_{0}-g_{1})(y))j_{1}}\dots \sum_{j_{t}\in \mathbb{F}_{p}}\zeta_{p}^{(i_{t}-(g_{0}-g_{t})(y))j_{t}}\\
    &=p^{-t}\sum_{i_{1}, \dots, i_{t} \in \mathbb{F}_{p}}\zeta_{p^k}^{g(i_{1}, \dots, i_{t})}W_{f_{(i_{1}, \dots, i_{t})}}(a)\sum_{j_{1}, \dots, j_{t}\in \mathbb{F}_{p}}\zeta_{p}^{i_{1}j_{1}+\dots+i_{t}j_{t}}W_{G_{(j_{1}, \dots, j_{t})}}(b)\\
    & =u(b)p^{\frac{m}{2}}p^{-t}\zeta_{p}^{g_{0}^{*}(b)}\sum_{i_{1}, \dots, i_{t} \in \mathbb{F}_{p}}\zeta_{p^k}^{g(i_{1}, \dots, i_{t})}W_{f_{(i_{1}, \dots, i_{t})}}(a)\sum_{j_{1}\in \mathbb{F}_{p}}\zeta_{p}^{(g_{1}^{*}(b)-g_{0}^{*}(b)+i_{1})j_{1}}\dots \sum_{j_{t}\in \mathbb{F}_{p}}\zeta_{p}^{(g_{t}^{*}(b)-g_{0}^{*}(b)+i_{t})j_{t}}\\
    & =u(b)p^{\frac{m}{2}}\zeta_{p}^{g_{0}^{*}(b)}\zeta_{p^k}^{g(g_{0}^{*}(b)-g_{1}^{*}(b), \dots, g_{0}^{*}(b)-g_{t}^{*}(b))}W_{f_{(g_{0}^{*}(b)-g_{1}^{*}(b), \dots, g_{0}^{*}(b)-g_{t}^{*}(b))}}(a),
  \end{split}
  \end{equation}
where the fifth equation is obtained by the properties of bent functions $g_{i} \ (0\leq i \leq t)$. By (25), we can see that $F: V_{r}\times V_{m}\rightarrow \mathbb{Z}_{p^k}$ is a generalized $s$-plateaued function if $f_{c}, c\in \mathbb{F}_{p}^{t}$ are generalized $s$-plateaued functions from $V_{r}$ to $\mathbb{Z}_{p^k}$.
\end{proof}

If $s=0$, then Theorem 5 can be used to construct (non)-weakly regular generalized bent functions and the duals can be given. The following corollary is an immediate consequence of Theorem 5 and its proof.

\begin{corollary}\label{Corollary3}
If $s=0$, then the function $F: V_{r}\times V_{m}\rightarrow \mathbb{Z}_{p^k}$ constructed by Theorem 5 is a generalized bent function and its dual $F^{*}(x, y)=f^{*}_{(g_{0}^{*}(y)-g_{1}^{*}(y), \dots, g_{0}^{*}(y)-g_{t}^{*}(y))}(x)+p^{k-1}g_{0}^{*}(y)+g(g_{0}^{*}(y)-g_{1}^{*}(y), \dots, g_{0}^{*}(y)-g_{t}^{*}(y))$.
Furthermore, $F$ is non-weakly regular if any one of the following conditions holds:

(1) There exists $i\in \mathbb{F}_{p}^{t}$ such that $f_{i}$ is non-weakly regular and $|\{b \in V_{m}: (g_{0}^{*}(b)-g_{1}^{*}(b), \dots, $ $g_{0}^{*}(b)-g_{t}^{*}(b))=i\}|\geq 1$;

(2) $u$ is a constant function and there exist $i_{1}\neq i_{2}\in \mathbb{F}_{p}^{t}$ such that $f_{i_{1}}, f_{i_{2}}$ are weakly regular with $\mu_{f_{i_{1}}}\neq \mu_{f_{i_{2}}}$ and $|\{b \in V_{m}: (g_{0}^{*}(b)-g_{1}^{*}(b), \dots, $ $g_{0}^{*}(b)-g_{t}^{*}(b))=i_{j}\}|\geq 1$ for $j=1, 2$;

(3) $u$ is not a constant function and $\mu_{f_{i}}=c, i \in \mathbb{F}_{p}^{t}$, where $c$ is a constant function independent of $i$.
\end{corollary}

Obviously arbitrary two Boolean bent functions $g_{0}, g_{1}$ satisfy the conditions of Theorem 5. When $p=2, k=t=1$, $f_{0}, f_{1}$ are Boolean plateaued functions, $g_{0}, g_{1}$ are Boolean bent functions and $g=0$, the Boolean plateaued function constructed by Theorem 5 is $F(x, y)=f_{g_{0}(y)+g_{1}(y)}(x)+g_{0}(y)=g_{0}(y)+f_{0}(x)+(f_{0}(x)+f_{1}(x))(g_{0}(y)+g_{1}(y))$. It is just the well-known indirect sum \cite{Carlet3}. Hence, Theorem 5 can be seen as an extension of the indirect sum. Also note that Theorem 4.2 (i) of \cite{Riera} for generalized Boolean plateaued functions as a generalization of the indirect sum is a special case of Theorem 5. When $p$ is an odd prime or $t\geq 2$, the conditions in Theorem 5 for bent functions $g_{i} \ (0\leq i \leq t)$ seem harsh. In the following, we illustrate that when $p$ is an odd prime or $t\geq 2$, although the conditions for used bent functions $g_{i} \ (0\leq i \leq t)$ seem harsh, we can still provide abundant bent functions that satisfy such conditions.

\begin{itemize}
  \item Let $g_{i} \ (0\leq i \leq t): \mathbb{F}_{p}^{m} \times \mathbb{F}_{p}^{m}\rightarrow \mathbb{F}_{p}$ be defined as
  \begin{equation}\label{26}
   g_{0}(y_{1}, y_{2})=y_{1}\cdot \pi(y_{2}), g_{i}(y_{1}, y_{2})=g_{0}(y_{1}, y_{2})+h_{i}(y_{2}), 1\leq i \leq t,
  \end{equation}
  where $\pi$ is a permutation over $\mathbb{F}_{p}^{m}$ and $h_{i} \ (1\leq i \leq t): \mathbb{F}_{p}^{m}\rightarrow \mathbb{F}_{p} $ are arbitrary functions. Then $g_{i} \ (0\leq i \leq t)$ are Maiorana-McFarland bent functions and by (7), it is easy to verify that $g_{i} \ (0\leq i \leq t)$ satisfy the conditions of Theorem 5.
  \item Let $g_{i} \ (0\leq i \leq t): \mathbb{F}_{p^m} \times \mathbb{F}_{p^m}\rightarrow \mathbb{F}_{p}$ be defined as
  \begin{equation}\label{27}
     g_{i}(y_{1}, y_{2})=Tr_{1}^{m}(\alpha_{i} G(y_{1}y_{2}^{p^m-2})),
  \end{equation}
  where $m\geq t+1$, $G$ is a permutation over $\mathbb{F}_{p^m}$ with $G(0)=0$ and $\alpha_{0}, \dots, \alpha_{t} \in \mathbb{F}_{p^m}$ are linearly independent over $\mathbb{F}_{p}$. Then $g_{i} \ (0\leq i \leq t)$ are $PS_{ap}$ bent functions and by (8), it is easy to verify that $g_{i} \ (0\leq i \leq t)$ satisfy the conditions of Theorem 5.
\end{itemize}

Since the bent functions $g_{i} \ (0\leq i \leq t)$ defined by (26) (respectively, (27)) satisfy the corresponding conditions in Theorem 5, we obtain the following constructions from Theorem 5.

\begin{corollary}\label{Corollary4}
Let $p$ be a prime, $k, t, r, m$ be positive integers, $s\ (\leq r)$ be a non-negative integer and $r+s$ be even for $p=2, k=1$. Let $f_{c} \ (c \in \mathbb{F}_{p}^{t}): V_{r} \rightarrow \mathbb{Z}_{p^k}$ be generalized $s$-plateaued functions, $g_{i} \ (0\leq i \leq t): \mathbb{F}_{p}^{m}\times \mathbb{F}_{p}^{m}\rightarrow \mathbb{F}_{p}$ be defined by (26), and $g: \mathbb{F}_{p}^{t}\rightarrow \mathbb{Z}_{p^k}$ be an arbitrary function. Then $F: V_{r}\times \mathbb{F}_{p}^{m}\times \mathbb{F}_{p}^{m}\rightarrow \mathbb{Z}_{p^k}$ defined as $F(x, y)=f_{(g_{0}(y)-g_{1}(y), \dots, g_{0}(y)-g_{t}(y))}(x)+p^{k-1}g_{0}(y)+g(g_{0}(y)-g_{1}(y), \dots, g_{0}(y)-g_{t}(y))$ is a generalized $s$-plateaued function.
\end{corollary}

\begin{corollary}\label{Corollary5}
Let $p$ be a prime, $k, t, r, m$ be positive integers with $m\geq t+1$, $s\ (\leq r)$ be a non-negative integer and $r+s$ be even for $p=2, k=1$. Let $f_{c} \ (c \in \mathbb{F}_{p}^{t}): V_{r} \rightarrow \mathbb{Z}_{p^k}$ be generalized $s$-plateaued functions, $g_{i} \ (0\leq i \leq t): \mathbb{F}_{p^m}\times \mathbb{F}_{p^m}\rightarrow \mathbb{F}_{p}$ be defined by (27), and $g: \mathbb{F}_{p}^{t}\rightarrow \mathbb{Z}_{p^k}$ be an arbitrary function. Then $F: V_{r}\times \mathbb{F}_{p^m}\times \mathbb{F}_{p^m}\rightarrow \mathbb{Z}_{p^k}$ defined as $F(x, y)=f_{(g_{0}(y)-g_{1}(y), \dots, g_{0}(y)-g_{t}(y))}(x)+p^{k-1}g_{0}(y)+g(g_{0}(y)-g_{1}(y), \dots, g_{0}(y)-g_{t}(y))$ is a generalized $s$-plateaued function.
\end{corollary}

By using Corollary 5, we show that the functions $g_{i} \ (0\leq i \leq t)$ defined below also satisfy the corresponding conditions in Theorem 5.
\begin{itemize}
  \item Let $g_{i} \ (0\leq i \leq t): V_{m}\times \mathbb{F}_{p^{t+2}}\times \mathbb{F}_{p^{t+2}}\rightarrow \mathbb{F}_{p}$ be defined as
  \begin{equation}\label{28}
    g_{i}(x, y_{1}, y_{2})=h_{Tr_{1}^{t+2}(\beta G(y_{1} y_{2}^{p^{t+2}-2}))}(x)+Tr_{1}^{t+2}(\alpha_{i} G(y_{1}y_{2}^{p^{t+2}-2})),
  \end{equation}
  where $m$ is even if $p=2$, $h_{c} \ (c \in \mathbb{F}_{p}): V_{m}\rightarrow \mathbb{F}_{p}$ are bent functions, $G$ is a permutation over $\mathbb{F}_{p^{t+2}}$ with $G(0)=0$ and $\beta, \alpha_{0}, \dots, \alpha_{t} \in \mathbb{F}_{p^{t+2}}$ are linearly independent over $\mathbb{F}_{p}$.
\end{itemize}
\begin{lemma}\label{Lemma1}
Let $g_{i} \ (0\leq i \leq t): V_{m}\times \mathbb{F}_{p^{t+2}}\times \mathbb{F}_{p^{t+2}}\rightarrow \mathbb{F}_{p}$ be defined by (28). Then for any $j=(j_{1}, \dots, j_{t})\in \mathbb{F}_{p}^{t}$, $G_{j}=(1-j_{1}-\dots-j_{t})g_{0}+j_{1}g_{1}+\dots+j_{t}g_{t}$ is a bent function and $G_{j}^{*}=(1-j_{1}-\dots-j_{t})g_{0}^{*}+j_{1}g_{1}^{*}+\dots+j_{t}g_{t}^{*}$ and $\mu_{G_{j}}=u$, where $\mu_{G_{j}}$ is defined by (3) and $u$ is a function from $V_{m}\times \mathbb{F}_{p^{t+2}}\times \mathbb{F}_{p^{t+2}}$ to $\{\pm 1, \pm \sqrt{-1}\}$ independent of $j$.
\end{lemma}
\begin{proof}
For any $0\leq i \leq t$, $g_{i}(x, y_{1}, y_{2})=h_{l_{0}(y_{1}, y_{2})-l_{1}(y_{1}, y_{2})}(x)+l_{0}(y_{1}, y_{2})$, where $l_{0}(y_{1}, y_{2})=Tr_{1}^{t+2}(\alpha_{i} G(y_{1}y_{2}^{p^{t+2}-2})), l_{1}(y_{1}, y_{2})=Tr_{1}^{t+2}((\alpha_{i}-\beta)G(y_{1}y_{2}^{p^{t+2}-2}))$. Since for any $0\leq i \leq t$, $\alpha_{i}, \alpha_{i}-\beta$ are linearly independent over $\mathbb{F}_{p}$, by Corollary 5, $g_{i} \ (0\leq i \leq t)$ defined by (28) are bent functions. Further by (8) and (25), we can see that $g_{i}^{*}(x, y_{1}, y_{2})=h_{Tr_{1}^{t+2}(\beta G(-y_{1}^{p^{t+2}-2}y_{2}))}^{*}(x)+Tr_{1}^{t+2}(\alpha_{i} G(-y_{1}^{p^{t+2}-2}y_{2}))$. For any $j=(j_{1}, \dots, j_{t})\in \mathbb{F}_{p}^{t}$, $G_{j}(x, y_{1}, y_{2})=h_{Tr_{1}^{t+2}(\beta G(y_{1} y_{2}^{p^{t+2}-2}))}(x)+Tr_{1}^{t+2}(\alpha G(y_{1}y_{2}^{p^{t+2}-2}))$, where $\alpha=(1-j_{1}-\dots-j_{t})\alpha_{0}+j_{1}\alpha_{1}+\dots+j_{t}\alpha_{t}$. Also by Corollary 5 and (8), (25), we can see that $G_{j}$ is a bent function with $G_{j}^{*}(x, y_{1}, y_{2})=h_{Tr_{1}^{t+2}(\beta G(-y_{1}^{p^{t+2}-2}y_{2}))}^{*}(x)+Tr_{1}^{t+2}(\alpha G(-y_{1}^{p^{t+2}-2}y_{2}))$ and $\mu_{G_{j}}(x, y_{1}, y_{2})=\mu_{h_{Tr_{1}^{t+2}(\beta G(-y_{1}^{p^{t+2}-2}y_{2}))}}(x)$, which is independent of $j$.
\end{proof}

Since the bent functions $g_{i} \ (0\leq i \leq t)$ defined by (28) satisfy the corresponding conditions in Theorem 5, we obtain the following construction from Theorem 5.

\begin{corollary}\label{Corollary5}
Let $p$ be a prime, $k, t, r, m$ be positive integers, $s\ (\leq r)$ be a non-negative integer and $m$ be even for $p=2$, $r+s$ be even for $p=2, k=1$. Let $f_{c} \ (c \in \mathbb{F}_{p}^{t}): V_{r} \rightarrow \mathbb{Z}_{p^k}$ be generalized $s$-plateaued functions, $g_{i} \ (0\leq i \leq t): V_{m}\times \mathbb{F}_{p^{t+2}}\times \mathbb{F}_{p^{t+2}}\rightarrow \mathbb{F}_{p}$ be defined by (28), and $g: \mathbb{F}_{p}^{t}\rightarrow \mathbb{Z}_{p^k}$ be an arbitrary function. Then $F(x, y)=f_{(g_{0}(y)-g_{1}(y), \dots, g_{0}(y)-g_{t}(y))}(x)+p^{k-1}g_{0}(y)+g(g_{0}(y)-g_{1}(y), \dots, g_{0}(y)-g_{t}(y))$ is a generalized $s$-plateaued function from $V_{r}\times V_{m}\times \mathbb{F}_{p^{t+2}}\times \mathbb{F}_{p^{t+2}}$ to $\mathbb{Z}_{p^k}$.
\end{corollary}

Let $s=0$, then the above corollaries provide constructions of generalized bent functions.

\begin{theorem}\label{Theorem6}
When $s=0$, then the function $F$ constructed by Corollary 4 (resp., Corollary 5, Corollary 6) is a generalized bent function and its dual $F^{*}(x, y)=f^{*}_{(g_{0}^{*}(y)-g_{1}^{*}(y), \dots, g_{0}^{*}(y)-g_{t}^{*}(y))}(x)+p^{k-1}g_{0}^{*}(y)+g(g_{0}^{*}(y)-g_{1}^{*}(y), \dots, $ $g_{0}^{*}(y)-g_{t}^{*}(y))$.
\end{theorem}
\begin{remark}\label{Remark4}
When $k=1$, Corollaries 4, 5, 6 are applicable for constructing $p$-ary plateaued functions for any prime $p$. When $s=0, k=1, m=t$, $g=0$ and bent functions $g_{i} \ (0\leq i \leq t)$ are defined by (26) with $h_{i}(y_{2})=-y_{2, i}, 1\leq i \leq t$, where $y_{2}=(y_{2, 1}, \dots, y_{2, t})\in \mathbb{F}_{p}^{t}$, then the corresponding construction in Corollary 4 is just the canonical way to construct the so-called Generalized Maiorana-McFarland bent functions given in (6) of \cite{Cesmelioglu3}. By Theorem 2 and its proof of \cite{Cesmelioglu3}, any bent function in the completed Generalized Maiorana-McFarland class (that is, equivalent to a Generalized Maiorana-McFarland bent function) is equivalent to an Maiorana-McFarland bent function or a bent function of the form (6) of \cite{Cesmelioglu3}. Hence, any bent function in the completed Generalized Maiorana-McFarland class and not in the completed Maiorana-McFarland class is equivalent to a bent function which can be constructed by the generalized indirect sum.
\end{remark}

We give some examples. The third example gives a non-weakly regular bent function which is not in the completed Generalized Maiorana-McFarland class.

\begin{example}\label{Example6}
Let $f_{c} \ (c \in \mathbb{F}_{5}^{2}): \mathbb{F}_{5}^{4}\rightarrow \mathbb{Z}_{5^3}$ be defined as $f_{(0, 1)}(x_{1}, \dots, x_{4})=5^{2}(x_{1}^{2}+x_{2}^{2})$ and $f_{c}(x_{1}, \dots, x_{4})=5^{2}(x_{1}^{2}+2x_{3}^{2})$ if $c\neq (0, 1)$. Then $f_{c} \ (c \in \mathbb{F}_{5}^{2})$ are (trivial) generalized $2$-plateaued functions. Let $g_{i} \ (0\leq i \leq 2): \mathbb{F}_{5}^{4}\rightarrow \mathbb{F}_{5}$ be defined as $g_{0}(x_{1}, \dots, x_{4})=x_{1}x_{3}+x_{2}x_{4}$, $g_{1}(x_{1}, \dots, x_{4})=x_{1}x_{3}+x_{2}x_{4}+4x_{3}^{3}$, $g_{2}(x_{1}, \dots, x_{4})=x_{1}x_{3}+x_{2}x_{4}+4x_{4}^{3}$, and $g: \mathbb{F}_{5}^{2}\rightarrow \mathbb{Z}_{5^3}$ be defined as $g(x_{1}, x_{2})=5x_{1}+x_{2}$. Then the function $F: \mathbb{F}_{5}^{8}\rightarrow \mathbb{Z}_{5^3}$ constructed by Corollary 4 is a (non-trivial) generalized $2$-plateaued function. One can verify that its Walsh support is not an affine subspace.
\end{example}

\begin{example}\label{Example7}
Let $f_{0}(x_{1}, x_{2}, x_{3})=7(x_{1}^{2}+x_{2}^{2})$, $f_{1}(x_{1}, x_{2}, x_{3})=7(x_{1}^{2}+3x_{2}^{2})$, $f_{2}(x_{1}, $ $x_{2}, x_{3})=7(x_{1}^{2}+2x_{3}^{2})$, $f_{3}(x_{1}, x_{2}, x_{3})=7(x_{1}^{2}+5x_{3}^{2})$, $f_{4}(x_{1}, x_{2}, x_{3})=7(x_{2}^{2}+4x_{3}^{2})$, $f_{5}(x_{1}, x_{2}, x_{3})=7(x_{2}^{2}+6x_{3}^{2})$, $f_{6}(x_{1}, x_{2}, x_{3})=7(x_{1}^{2}+3x_{2}^{2}+x_{3})$. Then $f_{i}\ (i\in \mathbb{F}_{7}): \mathbb{F}_{7}^{3} \rightarrow \mathbb{Z}_{7^2}$ are (trivial) generalized $1$-plateaued functions. Let $g_{0}(y_{1}, y_{2})=Tr_{1}^{2}(y_{1}y_{2}^{47}), g_{1}(y_{1}, y_{2})=Tr_{1}^{2}(zy_{1}y_{2}^{47}), (y_{1}, y_{2}) \in \mathbb{F}_{7^2}\times \mathbb{F}_{7^2}$, where $z$ is the primitive element of $\mathbb{F}_{7^2}$ with $z^2+6z+3=0$. Let $g: \mathbb{F}_{7}\rightarrow \mathbb{Z}_{7^2}$ be defined as $g(x)=x^{5}+2x^{3}$. Then the function $F: \mathbb{F}_{7}^{3}\times \mathbb{F}_{7^2}\times \mathbb{F}_{7^2}\rightarrow \mathbb{Z}_{7^2}$ constructed by Corollary 5 is a (non-trivial) generalized $1$-plateaued function and one can verify that the Walsh support is not an affine subspace.
\end{example}

\begin{example}\label{Example8}
Let $\xi$ be the primitive element of $\mathbb{F}_{3^4}$ with $\xi^{4}+2\xi^{3}+2=0$. Let $z$ be the primitive element of $\mathbb{F}_{3^2}$ with $z^2+z+2=0$. Let $f_{0}(x)=Tr_{1}^{4}(x^{34}+x^{2})$, $f_{1}(x)=Tr_{1}^{4}(x^{2})$, $f_{2}(x)=Tr_{1}^{4}(\xi x^{2})$, $x \in \mathbb{F}_{3^4}$. Then $f_{0}, f_{1}, f_{2}$ are weakly regular bent functions with $\mu_{f_{0}}=\mu_{f_{1}}=-1, \mu_{f_{2}}=1$. Let $g_{0}(y_{1}, y_{2})=Tr_{1}^{2}(y_{1}y_{2}^{7})$, $g_{1}(y_{1}, y_{2})=Tr_{1}^{2}(z y_{1}y_{2}^{7})$, $(y_{1}, y_{2})\in \mathbb{F}_{3^2}\times \mathbb{F}_{3^2}$. Let $g=0$. Then the function $F: \mathbb{F}_{3^4}\times \mathbb{F}_{3^2}\times \mathbb{F}_{3^2}\rightarrow \mathbb{F}_{3}$ constructed by Corollary 5 is a non-weakly regular bent function. Further, we will prove in Appendix that it is not in the completed Generalized Maiorana-McFarland class.
\end{example}

In the rest of this section, by using Corollary 5, we give constructions of plateaued functions in the subclass \emph{WRP} of the class of weakly regular plateaued functions and vectorial plateaued functions.

In \cite{Mesnager7}, Mesnager and S{\i}nak introduced the notion of class \emph{WRP}, which is a subclass of the class of weakly regular plateaued functions and plays an important role in constructing minimal linear codes and strongly regular graphs (see \cite{Mesnager7}, \cite{Mesnager8}).
Let $p$ be an odd prime. Let $f: V_{n}\rightarrow \mathbb{F}_{p}$ be an unbalanced weakly regular $s$-plateaued function. If $f(0)=0$ and there exists an even positive integer $h$ with $gcd(h-1, p-1)=1$ such that $f(ax)=a^{h}f(x), x \in V_{n}$ for any $a\in \mathbb{F}_{p}^{*}$, then $f$ belongs to the class \emph{WRP}. Note that all quadratic functions without affine term are in the class \emph{WRP} and $h=2$. We give a construction of non-quadratic plateaued functions in the class \emph{WRP} by using Corollary 5.

Let $p$ be an odd prime and $m$ be an even positive integer. Let $f: \mathbb{F}_{p}^{m}\rightarrow \mathbb{F}_{p}$ be a partial spread bent function (see \cite{Lisonek}). Then by Theorem 3.3 and Theorem 3.6 of \cite{Lisonek}, it is easy to see that for any $a\in \mathbb{F}_{p}^{*}$, $f(ax)=f(x)$. Let $t, r$ be positive integers, $s$ be a non-negative integer and $r-s$ be an even positive integer. For any $i\in \mathbb{F}_{p}^{t}$, let $b_{i}: \mathbb{F}_{p}^{r-s}\rightarrow \mathbb{F}_{p}$ be a partial spread bent function, $M_{i}$ be a row full rank matrix over $\mathbb{F}_{p}$ of size $(r-s) \times r$. Define
\begin{equation}\label{29}
 f_{i}(x)=b_{i}(xM_{i}^{T}), x \in \mathbb{F}_{p}^{r}, i \in \mathbb{F}_{p}^{t}.
\end{equation}
Then for any $i\in \mathbb{F}_{p}^{t}$, $f_{i}$ is an $s$-plateaued function whose Walsh support $S_{f_{i}}$ is a linear subspace and $\mu_{f_{i}}(x)=1, x \in S_{f_{i}}$ by Theorem 1. Further, for any $a\in \mathbb{F}_{p}^{*}$, $f_{i}(ax)=f_{i}(x), x \in \mathbb{F}_{p}^{r}$.

\begin{proposition}\label{Proposition3}
Let $p$ be an odd prime and $k=1$. Let $t, r, m$ be positive integers with $m\geq t+1$. Let $s\ (< r)$ be a non-negative integer. Let $g_{j} \ (0\leq j \leq t) : \mathbb{F}_{p^m}\times \mathbb{F}_{p^m}\rightarrow \mathbb{F}_{p}$ be defined as $g_{j}(y)=Tr_{1}^{m}(\alpha_{j} G(y_{1}y_{2}^{p^m-2})), y=(y_{1}, y_{2})\in \mathbb{F}_{p^m}\times \mathbb{F}_{p^m}$, where $G$ is a permutation over $\mathbb{F}_{p^m}$ with $G(0)=0$ and $\alpha_{0}, \alpha_{1}, \dots, \alpha_{t}\in \mathbb{F}_{p^m}$ are linearly independent over $\mathbb{F}_{p}$.
\begin{itemize}
  \item Case $p=3$: Let $f_{i} \ (i \in \mathbb{F}_{p}^{t}) : V_{r}\rightarrow \mathbb{F}_{p}$ be weakly regular $s$-plateaued functions satisfying $\mu_{f_{i}}(x)=u, x \in S_{f_{i}}, i \in \mathbb{F}_{p}^{t}$, where $\mu_{f_{i}}$ is defined by (3) and $u$ is some constant independent of $i$, $f_{i}(ax)=a^{2}f_{i}(x), x \in V_{r}, i \in \mathbb{F}_{p}^{t}$ for any $a \in \mathbb{F}_{p}^{*}$ and $0 \in S_{f_{(0, \dots, 0)}}$. Let $g: \mathbb{F}_{p}^{t}\rightarrow \mathbb{F}_{p}$ be an arbitrary function with $g(0, \dots, 0)=-f_{(0, \dots, 0)}(0)$.
  \item Case $p\geq 5$ : Let $r-s$ be even. Let $f_{i}, i \in \mathbb{F}_{p}^{t}$ be defined as (29). Let $g: \mathbb{F}_{p}^{t}\rightarrow \mathbb{F}_{p}$ be an arbitrary function with $g(0, \dots, 0)=-f_{(0, \dots, 0)}(0)$.
\end{itemize}
Then the function $F: V_{r}\times \mathbb{F}_{p^m}\times \mathbb{F}_{p^m}\rightarrow \mathbb{F}_{p}$ constructed by Corollary 5 is a weakly regular $s$-plateaued function and in the class WRP.
\end{proposition}
\begin{proof}
By Corollary 5, $F$ is an $s$-plateaued function. Further, by the proof of Theorem 5, it is easy to see that $F$ is weakly regular and
$S_{F}=\cup_{y \in \mathbb{F}_{p^m}\times \mathbb{F}_{p^m}}S_{f_{(g_{0}^{*}(y)-g_{1}^{*}(y), \dots, g_{0}^{*}(y)-g_{t}^{*}(y))}}\times \{y\}$. Since $g_{0}^{*}(0, 0)-g_{j}^{*}(0, 0)=0, 1\leq j \leq t$ and $0 \in S_{f_{(0, \dots, 0)}}$, we have $(0, 0, 0) \in S_{F}$, that is, $F$ is unbalanced. Since $g(0, \dots, 0)=-f_{(0, \dots, 0)}(0)$, $F(0, 0, 0)=0$. As $f_{i}(ax)=f_{i}(x), x \in V_{r}, i\in \mathbb{F}_{p}^{t}$, $g_{j}(ay)=g_{j}(y), y\in \mathbb{F}_{p^m}\times \mathbb{F}_{p^m}, 0\leq j\leq t$ for any $a \in \mathbb{F}_{p}^{*}$, the weakly regular plateaued function $F$ constructed by Corollary 5 satisfies $F(ax, ay)=F(x, y)=a^{p-1}F(x, y), (x, y)\in V_{r} \times \mathbb{F}_{p^m} \times \mathbb{F}_{p^m}$ for any $a \in \mathbb{F}_{p}^{*}$. Note that $p-1$ is even and $gcd(p-2, p-1)=1$. By definition, $F$ is in the \emph{WRP} class.
\end{proof}

We give an example of non-quadratic plateaued function in the \emph{WRP} class by using Proposition 3.
\begin{example}
Let $p=3, t=1, r=2, m=2, s=1$. Let $z$ be the primitive element of $\mathbb{F}_{3^2}$ with $z^2+2z+2=0$. Let $g_{0}(y_{1}, y_{2})=Tr_{1}^{2}(y_{1}y_{2}^{7}), g_{1}(y_{1}, y_{2})=Tr_{1}^{2}(zy_{1}y_{2}^{7}), (y_{1}, y_{2})\in \mathbb{F}_{3^2}\times \mathbb{F}_{3^2}$. Let $f_{0}(x_{1}, x_{2})=x_{1}^{2}, f_{1}(x_{1}, x_{2})=x_{2}^{2}, f_{2}(x_{1}, x_{2})$ $=x_{1}^{2}+x_{1}x_{2}+x_{2}^{2}, (x_{1}, x_{2})\in \mathbb{F}_{3}^{2}$. Then $f_{i}, i \in \mathbb{F}_{3}$ are $1$-plateaued functions with $\mu_{f_{i}}(x_{1}, x_{2})=\sqrt{-1}, (x_{1}, x_{2})\in S_{f_{i}}$ and $f_{i}(ax_{1}, ax_{2})=a^{2}f_{i}(x_{1}, x_{2}), (x_{1}, x_{2})\in \mathbb{F}_{3}^{2}$ for any $a\in \mathbb{F}_{3}^{*}$ and $(0, 0)\in S_{f_{0}}$. Let $g=0$. Then the function $F$ constructed by Proposition 3 is $F(x_{1},x_{2}, y_{1}, y_{2})=Tr_{1}^{2}(y_{1}y_{2}^{7})+x_{1}^{2}+(Tr_{1}^{2}((1-z)y_{1}y_{2}^{7}))^{2}(x_{1}^{2}+2x_{1}x_{2}+x_{2}^{2})+(Tr_{1}^{2}((1-z)y_{1}y_{2}^{7}))(x_{1}^{2}+x_{1}x_{2}), (x_{1}, x_{2}, y_{1}, y_{2})\in \mathbb{F}_{3}^{2}\times \mathbb{F}_{3^2}\times \mathbb{F}_{3^2}$, which is a non-quadratic weakly regular $1$-plateaued function and in the WRP class. Furthermore, one can verify that the Walsh support of $F$ is not an affine subspace, that is, $F$ is not a partially bent function.
\end{example}

Let $f=(f_{1}, \dots, f_{m})$ be a vectorial function from $V_{n}$ to $\mathbb{F}_{p}^{m}$. Then $f$ is said to be a vectorial plateaued function if for any nonzero vector $(c_{1}, \dots, c_{m})\in \mathbb{F}_{p}^{m}$, $\sum_{i=1}^{m}c_{i}f_{i}$ is a plateaued function from $V_{n}$ to $\mathbb{F}_{p}$. We give a construction of vectorial plateaued functions by using Corollary 5.

\begin{proposition}\label{Proposition4}
Let $p$ be a prime, $r\geq 1, m\geq 3, 0\leq s \leq r$ be integers and $r+s$ be even for $p=2$. Let $\{\alpha_{0}, \dots, \alpha_{m-1}\}$ be a basis of $\mathbb{F}_{p^m}$ over $\mathbb{F}_{p}$. Let $f_{0}, \dots, f_{p-1}: V_{r} \rightarrow \mathbb{F}_{p}$ be $s$-plateaued functions, $G$ be a permutation over $\mathbb{F}_{p^m}$ with $G(0)=0$. Define $h_{i}(x, y_{1}, y_{2})=f_{Tr_{1}^{m}(\alpha_{0} G(y_{1} y_{2}^{p^m-2}))}(x)+Tr_{1}^{m}(\alpha_{i} G(y_{1}y_{2}^{p^m-2})), (x, y_{1}, y_{2})\in V_{r}\times \mathbb{F}_{p^m}\times \mathbb{F}_{p^m}, 1\leq i \leq m-1$. Then vectorial function $H=(h_{1}, \dots, h_{m-1})$ is a vectorial plateaued function from $V_{r}\times \mathbb{F}_{p^m}\times \mathbb{F}_{p^m}$ to $\mathbb{F}_{p}^{m-1}$.
\end{proposition}
\begin{proof}
First by the similar argument as in the proof of Lemma 1, we have that if $\alpha, \beta \in \mathbb{F}_{p^m}$ are linearly independent over $\mathbb{F}_{p}$, then $h(x, y_{1}, y_{2})=f_{Tr_{1}^{m}(\beta G(y_{1}y_{2}^{p^m-2}))}(x)+Tr_{1}^{m}(\alpha G(y_{1}y_{2}^{p^m-2})), $ $(x, y_{1}, y_{2})\in V_{r}\times\mathbb{F}_{p^m}\times \mathbb{F}_{p^m}$ is an $s$-plateaued function, where $f_{0}, \dots, f_{p-1}$ are $s$-plateaued functions and $G$ is a permutation over $\mathbb{F}_{p^m}$ with $G(0)=0$.

For any nonzero vector $a=(a_{1}, \dots, a_{m-1})\in \mathbb{F}_{p}^{m-1}$, let $\bar{a}=\sum_{i=1}^{m-1}a_{i}, \alpha_{a}=\sum_{i=1}^{m-1}a_{i}\alpha_{i}$. If $\bar{a}\neq 0$, in this case $\sum_{i=1}^{m-1}a_{i}h_{i}(x, y_{1}, y_{2})=\bar{a}f_{Tr_{1}^{m}(\alpha_{0} G(y_{1}y_{2}^{p^m-2}))}(x)+Tr_{1}^{m}(\alpha_{a}G(y_{1}y_{2}^{p^m-2}))$. By Theorem 1 of  \cite{Cesmelioglu2}, $\bar{a}f_{0}, \dots, \bar{a}f_{p-1}$ are $s$-plateaued functions. Since $\bar{a}f_{0}, \dots, \bar{a}f_{p-1}$ are $s$-plateaued functions and $\alpha_{0}, \alpha_{a}$ are linearly independent, we have $\sum_{i=1}^{m-1}a_{i}h_{i}$ is an $s$-plateaued function. If $\bar{a}=0$, in this case $\sum_{i=1}^{m-1}a_{i}h_{i}(x, y_{1}, y_{2})=Tr_{1}^{m}(\alpha_{a}G(y_{1}y_{2}^{p^m-2}))$. Since $\alpha_{a} \neq 0$, it is easy to see that $\sum_{i=1}^{m-1}a_{i}h_{i}$ is an $r$-plateaued function.
\end{proof}

We give an example of vectorial plateaued function by using Proposition 4.

\begin{example}
Let $p=3, r=3, m=4, s=0$. Let $f_{j}(x)=Tr_{1}^{3}(\xi^{j} x^{2}), x \in \mathbb{F}_{3^3}, j \in \mathbb{F}_{3}$, where $\xi$ is a primitive element of $\mathbb{F}_{3^3}$. Then $f_{j}\ (j \in \mathbb{F}_{3})$ are weakly regular bent functions with $\mu_{f_{0}}=\mu_{f_{2}}=-\sqrt{-1}, \mu_{f_{1}}=\sqrt{-1}$. Let $h_{i}(x, y_{1}, y_{2})=f_{Tr_{1}^{4}(y_{1}y_{2}^{79})}(x)+Tr_{1}^{4}(z^{i} y_{1}y_{2}^{79}), (x, y_{1}, y_{2})\in \mathbb{F}_{3^3}\times \mathbb{F}_{3^4}\times \mathbb{F}_{3^4}, i=1, 2, 3$, where $z$ is a primitive element of $\mathbb{F}_{3^4}$. Then $H=(h_{1}, h_{2}, h_{3})$ is a vectorial plateaued function. Furthermore, one can verify that $H$ contains non-weakly regular plateaued component functions and weakly regular plateaued component functions.
\end{example}

\section{Constructions of pairwise disjoint spectra generalized plateaued functions}
In this section, we discuss the constructions of pairwise disjoint spectra generalized plateaued functions with (non)-affine Walsh supports and we present a construction of generalized bent functions by using pairwise disjoint spectra generalized plateaued functions as building blocks.

First of all, we illustrate that by using Theorem 1 and some known generalized bent functions as building blocks, we can construct pairwise disjoint spectra generalized plateaued functions.
\begin{itemize}
  \item Let $p$ be a prime, $n, s \ (<n)$ be positive integers and $n+s$ be even for $p=2$. Let $E$ and $E'$ be $(n-s)$-dimensional and $s$-dimensional linear subspaces of $\mathbb{F}_{p}^{n}$ respectively and satisfy $E\oplus E'=\mathbb{F}_{p}^{n}$, where $\oplus$ denotes direct sum. Note that this can be easily done, for example, let $E=<\alpha_{1}, \dots, \alpha_{n-s}>$ and $E'=<\alpha_{n-s+1}, \dots, \alpha_{n}>$, where $\{\alpha_{1}, \dots, \alpha_{n}\}$ is some basis of $\mathbb{F}_{p}^{n}$. Suppose $E'=\{e'_{0}, \dots, e'_{p^s-1}\}$. Let $S_{i}=e'_{i}+E$, $0\leq i \leq p^{s}-1$. Then $S_{i}\cap S_{j}=\emptyset$ if $i\neq j$. For any $0\leq i \leq p^s-1$, let $M_{i}$ be a matrix whose row vectors form a basis of the $(n-s)$-dimensional linear subspace $E$, and $g_{i} \ (0\leq i \leq p^s-1): \mathbb{F}_{p}^{n-s}\rightarrow \mathbb{Z}_{p^k}$ be generalized bent functions. Then by Theorem 1, for any $0\leq i \leq p^s-1$, $f_{i}: \mathbb{F}_{p}^{n}\rightarrow \mathbb{Z}_{p^k}$ defined as
   \begin{equation}\label{30}
   f_{i}(x)=g_{i}(xM_{i}^{T})+p^{k-1}e'_{i}\cdot x
   \end{equation}
   is a generalized $s$-plateaued function with $S_{i}$ as Walsh support. Therefore, $f_{i} \ (0\leq i \leq p^s-1)$ defined by (30) are pairwise disjoint spectra generalized $s$-plateaued functions.
\end{itemize}

We give an example to illustrate the above construction of pairwise disjoint spectra generalized plateaued functions.

\begin{example}\label{Example11}
Let $p=2, n=5, s=1, k=3$. Let $E=\mathbb{F}_{2}^{4} \times \{0\}$, $E'=\{e'_{0}=(0, 0, 0, 0, 0), e'_{1}=(0, 0, 0, 0, 1)\}$, then $E$ and $E'$ are $4$-dimensional and $1$-dimensional linear subspaces of $\mathbb{F}_{2}^{5}$ respectively and $E\oplus E'=\mathbb{F}_{2}^{5}$. Define $g_{0}, g_{1}: \mathbb{F}_{2}^{4}\rightarrow \mathbb{Z}_{2^3}$ as $g_{0}(x_{1}, \dots, x_{4})=4(x_{1}x_{3}+x_{2}x_{4})+2x_{3}+x_{3}x_{4}$, $g_{1}(x_{1}, \dots, x_{4})=4(x_{1}x_{3}+x_{2}x_{4})+2x_{1}x_{2}+x_{1}$, then $g_{0}, g_{1}$ are generalized bent functions. Let $M_{1}=M_{2}=M$, where $(1, 0, 0, 0, 0), (0, 1, 0, 0, 0), (0, 0, 1, 0, 0), $ $(0, 0, 0, 1, 0)$ are the row vectors of $M$. Then by (30), $f_{0}(x_{1}, \dots, x_{5})$ $=4(x_{1}x_{3}+x_{2}x_{4})+2x_{3}+x_{3}x_{4}$, $f_{1}(x_{1}, \dots, x_{5})=4(x_{1}x_{3}+x_{2}x_{4}+x_{5})+2x_{1}x_{2}+x_{1}$ are disjoint spectra generalized $1$-plateaued functions with $S_{f_{0}}=e'_{0}+E, S_{f_{1}}=e'_{1}+E$.
\end{example}

The Walsh supports of $f_{i} \ (0\leq i \leq p^s-1)$ defined by (30) are affine subspaces. In the following, we discuss the constructions of pairwise disjoint spectra (generalized) $p$-ary plateaued functions with (non)-affine Walsh supports for any prime $p$.

Let $f^{[i]} \ (0\leq i \leq p^s-1)$ be (generalized) $s$-plateaued functions with (the matrix form) $S_{f^{[i]}}=(T_{t_{1}^{[i]}}, \dots, T_{t_{s}^{[i]}}, T_{h_{1}^{[i]}}, \dots, T_{h_{n-s}^{[i]}})$ constructed by Theorem 2 (or Theorem 3, Theorem 4) for which
\begin{equation}\label{31}
  (t_{1}^{[i]}(x), \dots, t_{s}^{[i]}(x))=(t_{1}^{[0]}(x), \dots, t_{s}^{[0]}(x))+v_{i}, 0\leq i \leq p^s-1,
\end{equation}
and there exist $G_{j} \ (1\leq j\leq n-s): \mathbb{F}_{p}^{s}\rightarrow \mathbb{F}_{p}$ independent of $i$ such that
\begin{equation}\label{32}
  h_{j}^{[i]}(x)=G_{j}(t_{1}^{[i]}(x), \dots, t_{s}^{[i]}(x))+L_{j}(x)+b_{j}, 1\leq j \leq n-s, 0\leq i \leq p^s-1,
\end{equation}
where $\{v_{0}, \dots, v_{p^s-1}\}$ is the lexicographic order of $\mathbb{F}_{p}^{s}$, $L_{j} \ (1\leq j \leq n-s): V_{n-s}\rightarrow \mathbb{F}_{p}$ are linearly independent linear functions independent of $i$, $b_{j} \ (1\leq j\leq n-s)$ are arbitrary elements in $\mathbb{F}_{p}$ independent of $i$. By the construction given in Theorem 2 (or Theorem 3, Theorem 4), we can see that the conditions (31) and (32) for $f^{[i]} \ (0\leq i \leq p^s-1)$ are easy to be satisfied. We show that the (generalized) $s$-plateaued functions constructed by Theorem 2 (or Theorem 3, Theorem 4) which satisfy the conditions (31) and (32) are pairwise disjoint spectra (generalized) $s$-plateaued functions.
\begin{proposition}\label{Proposition5}
Let $f^{[i]} \ (0\leq i \leq p^s-1)$ be (generalized) $s$-plateaued functions with (the matrix form) $S_{f^{[i]}}=(T_{t_{1}^{[i]}}, \dots, T_{t_{s}^{[i]}}, T_{h_{1}^{[i]}}, \dots, T_{h_{n-s}^{[i]}})$ constructed by Theorem 2 (or Theorem 3, Theorem 4) for which the conditions (31) and (32) hold, then $f^{[i]} \ (0\leq i \leq p^s-1)$ are pairwise disjoint spectra (generalized) $s$-plateaued functions.
\end{proposition}
\begin{proof}
If there exist $0\leq i\neq i'\leq p^s-1$ and $x, x'\in V_{n-s}$ such that $(t_{1}^{[i]}(x), \dots, t_{s}^{[i]}(x), h_{1}^{[i]}(x), $ $\dots, h_{n-s}^{[i]}(x))=(t_{1}^{[i']}(x'), \dots, t_{s}^{[i']}(x'), h_{1}^{[i']}(x'), \dots, h_{n-s}^{[i']}(x'))$, then by (32), $L_{j}(x)=L_{j}(x')$ for any $1\leq j \leq n-s$. Since $L_{1}, \dots, L_{n-s}$ are linearly independent linear functions, we can see that $x=x'$ and thus $(t_{1}^{[i]}(x), \dots, $ $t_{s}^{[i]}(x))=(t_{1}^{[i']}(x), \dots, t_{s}^{[i']}(x))$. By (31), $(t_{1}^{[i]}(x), \dots, t_{s}^{[i]}(x))=(t_{1}^{[i']}(x), \dots,$ $t_{s}^{[i']}(x))+v_{i}-v_{i'}$, which contradicts $(t_{1}^{[i]}(x), \dots, t_{s}^{[i]}(x))=(t_{1}^{[i']}(x), \dots, t_{s}^{[i']}(x))$. Hence, $S_{f^{[i]}}\cap S_{f^{[i']}}=\emptyset$ for any $i\neq i'$.
\end{proof}

We give an example of pairwise disjoint spectra $3$-ary plateaued functions with non-affine Walsh supports by using Theorem 4 and Proposition 5.

\begin{example}\label{Example12}
Let $z$ be a primitive element of $\mathbb{F}_{3^4}$. Define $g^{[i]}=(g_{1}^{[i]}, g_{2}): \mathbb{F}_{3^4}\rightarrow \mathbb{F}_{3}^{2}, i \in \mathbb{F}_{3}$ as $g_{1}^{[i]}(x)=Tr_{1}^{4}(\alpha_{i}x^{2})$, $g_{2}(x)=Tr_{1}^{4}(zx^{2})$, where $\alpha_{0}=z^{11}, \alpha_{1}=z^{21}, \alpha_{2}=z^{31}$. Then by (9), it is easy to verify that $g^{[i]}=(g_{1}^{[i]}, g_{2}), i \in \mathbb{F}_{3}$ are vectorial bent functions with $\mu_{g_{1}^{[i]}+cg_{2}}=1$ for any $c \in \mathbb{F}_{3}$. For any $i \in \mathbb{F}_{3}$, let $t_{1}^{[i]}(x)=2g_{2}(x)+i$, $h_{1}^{[i]}(x)=t_{1}^{[i]}(x)+Tr_{1}^{4}(x)$, $h_{2}^{[i]}(x)=2t_{1}^{[i]}(x)+Tr_{1}^{4}(zx)$, $h_{3}^{[i]}(x)=t_{1}^{[i]}(x)+Tr_{1}^{4}(z^{2}x)$, $h_{4}^{[i]}(x)=2t_{1}^{[i]}(x)+Tr_{1}^{4}(z^{3}x)$. Then by (9), the $1$-plateaued functions $f^{[i]} \ (i \in \mathbb{F}_{3}): \mathbb{F}_{3}^{5}\rightarrow \mathbb{F}_{3}$ constructed by Theorem 4 are $f^{[i]}(x_{1}, \dots, x_{5})=Tr_{1}^{4}(\frac{2(x_{2}+x_{3}z+x_{4}z^{2}+x_{5}z^{3})^{2}}{\alpha_{i}+(2x_{1}+2x_{2}+x_{3}+2x_{4}+x_{5})z})+(x_{1}+x_{2}+2x_{3}+x_{4}+2x_{5})i$ with (the matrix form) $S_{f^{[i]}}=(T_{t_{1}^{[i]}}, T_{h_{1}^{[i]}}, \dots, T_{h_{4}^{[i]}})$. Since $t_{1}^{[i]} \ (i\in \mathbb{F}_{3})$ are all neither balanced nor constant, $S_{f^{[i]}} \ (i \in \mathbb{F}_{3})$ are all non-affine. Further, by Proposition 5, $S_{f^{[i]}} \cap S_{f^{[j]}}=\emptyset$ for any $i\neq j$, therefore, $f^{[i]} \ (i \in \mathbb{F}_{3})$ are pairwise disjoint spectra $1$-plateaued functions with non-affine Walsh supports.
\end{example}

Based on Corollaries 4, 5, 6, we give explicit constructions of pairwise disjoint spectra generalized plateaued functions whose Walsh supports can be non-affine. First we give a lemma.

\begin{lemma}\label{Lemma2}
Let $p$ be a prime, $r, t \ (< r), k$ be positive integers and $r+t$ be even when $p=2, k=1$. Suppose $l_{c} \ (c \in \mathbb{F}_{p}^{t}): V_{r}\rightarrow \mathbb{Z}_{p^k}$ are pairwise disjoint spectra generalized $t$-plateaued functions. Let $g_{i} \ (0\leq i \leq t)$ be defined by (26) (respectively, (27), (28)), and let $g^{[j]} \ (j \in \mathbb{F}_{p}^{t}): \mathbb{F}_{p}^{t}\rightarrow \mathbb{Z}_{p^k}$ be arbitrary functions. Define $f_{c}^{[j]} \ (c, j \in \mathbb{F}_{p}^{t}): V_{r}\rightarrow \mathbb{Z}_{p^k}$ as $f_{c}^{[j]}(x)=l_{P_{1}(c)+P_{2}(j)}(x)$, where $P_{1}: \mathbb{F}_{p}^{t}\rightarrow \mathbb{F}_{p}^{t}$ is an arbitrary function and $P_{2}$ is a permutation of $\mathbb{F}_{p}^{t}$. Then $F^{[j]}(x, y)=f^{[j]}_{(g_{0}(y)-g_{1}(y), \dots, g_{0}(y)-g_{t}(y))}(x)+p^{k-1}g_{0}(y)+g^{[j]}(g_{0}(y)-g_{1}(y), \dots, g_{0}(y)-g_{t}(y))$, $j \in \mathbb{F}_{p}^{t}$ are pairwise disjoint spectra generalized $t$-plateaued functions.
\end{lemma}
\begin{proof}
For the sake of simplicity, we only consider the case that $g_{i} \ (0\leq i \leq t)$ are defined by (26) since the other cases are the same. First of all, by Corollary 4, $F^{[j]} \ (j \in \mathbb{F}_{p}^{t})$ are generalized $t$-plateaued functions. Suppose there exist $j, j' \in \mathbb{F}_{p}^{t}$ with $j \neq j'$ and $a \in V_{r}, b \in \mathbb{F}_{p}^{m}\times \mathbb{F}_{p}^{m}$ such that $(a, b)\in S_{F^{[j]}}$ and $(a, b)\in S_{F^{[j']}}$. By (25), we have $S_{F^{[j]}}=\cup_{y\in \mathbb{F}_{p}^{m}\times\mathbb{F}_{p}^{m}}S_{f^{[j]}_{(g_{0}^{*}(y)-g_{1}^{*}(y), \dots, g_{0}^{*}(y)-g_{t}^{*}(y))}}\times \{y\}$. Then $a \in S_{f^{[j]}_{c}}$ and $a \in S_{f^{[j']}_{c}}$ for $c=(g_{0}^{*}(b)-g_{1}^{*}(b), \dots, g_{0}^{*}(b)-g_{t}^{*}(b))$. By the definition of $f^{[j]}_{c}, f^{[j']}_{c}$, we have $f^{[j]}_{c}=l_{P_{1}(c)+P_{2}(j)}$ and $f^{[j']}_{c}=l_{P_{1}(c)+P_{2}(j')}$, therefore, $a \in S_{l_{P_{1}(c)+P_{2}(j)}}$
and $a \in S_{l_{P_{1}(c)+P_{2}(j')}}$, which is a contradiction since $l_{P_{1}(c)+P_{2}(j)}$ and $l_{P_{1}(c)+P_{2}(j')}$ are disjoint spectra functions.
\end{proof}

By Lemma 2, we obtain the following explicit constructions of pairwise disjoint spectra generalized plateaued functions whose Walsh supports can be non-affine.

\begin{theorem}\label{Theorem7}
Let $p$ be a prime, $r, t \ (< r), k$ be positive integers and $r+t$ be even when $p=2$. Let $\beta_{c}, c \in \mathbb{F}_{p}^{t}$ be all the elements in $\mathbb{F}_{p^t}$. For any $c \in \mathbb{F}_{p}^{t}$, let $\gamma_{c} \in \mathbb{F}_{p^{r-t}}^{*}$ and $\gamma_{c} \notin \{x^{3}: x \in \mathbb{F}_{p^{r-t}}\}$ when $p=2$. Define $l_{c} \ (c \in \mathbb{F}_{p}^{t}): \mathbb{F}_{p^{r-t}}\times \mathbb{F}_{p^t}\rightarrow \mathbb{Z}_{p^k}$ as $l_{c}(x)=p^{k-1}(Tr_{1}^{r-t}(\gamma_{c}x_{1}^{d})+Tr_{1}^{t}(\beta_{c}x_{2})), x=(x_{1}, x_{2})\in  \mathbb{F}_{p^{r-t}}\times \mathbb{F}_{p^t}$, where $d=2$ when $p$ is an odd prime, and $d=3$ when $p=2$. Let $g_{i} \ (0\leq i \leq t)$ be defined by (26) (respectively, (27), (28)), and let $g^{[j]} \ (j \in \mathbb{F}_{p}^{t}): \mathbb{F}_{p}^{t}\rightarrow \mathbb{Z}_{p^k}$ be arbitrary functions. Define $f_{c}^{[j]} \ (c, j \in \mathbb{F}_{p}^{t}): \mathbb{F}_{p^{r-t}}\times \mathbb{F}_{p^t} \rightarrow \mathbb{Z}_{p^k}$ as $f_{c}^{[j]}(x)=l_{P_{1}(c)+P_{2}(j)}(x)$, where $P_{1}: \mathbb{F}_{p}^{t}\rightarrow \mathbb{F}_{p}^{t}$ is an arbitrary function and $P_{2}$ is a permutation of $\mathbb{F}_{p}^{t}$. Then $F^{[j]}(x, y)=f^{[j]}_{(g_{0}(y)-g_{1}(y), \dots, g_{0}(y)-g_{t}(y))}(x)+p^{k-1}g_{0}(y)+g^{[j]}(g_{0}(y)-g_{1}(y), \dots, g_{0}(y)-g_{t}(y))$, $j \in \mathbb{F}_{p}^{t}$ are pairwise disjoint spectra generalized $t$-plateaued functions.
\end{theorem}
\begin{proof}
When $p$ is an odd prime, by (9), for any $c \in \mathbb{F}_{p}^{t}$, $Tr_{1}^{r-t}(\gamma_{c}x^{2})$ is a bent function. When $p=2$, for any $c \in \mathbb{F}_{p}^{t}$, $Tr_{1}^{r-t}(\gamma_{c}x^{3})$ is a Gold bent function (see \cite{Dillon}). Therefore, it is easy to see that $l_{c} \ (c \in \mathbb{F}_{p}^{t})$ are (trivial) generalized $t$-plateaued functions with $S_{l_{c}}=\mathbb{F}_{p^{r-t}}\times \{\beta_{c}\}$. Obviously $S_{l_{c}}\cap S_{l_{c'}}=\emptyset$ if $c\neq c'$, that is, $l_{c} \ (c \in \mathbb{F}_{p}^{t})$ are pairwise disjoint spectra generalized $t$-plateaued functions. By Lemma 2, the theorem holds.
\end{proof}

\begin{remark}\label{Remark5}
For a function $f: V_{n}\rightarrow \mathbb{Z}_{p^k}$ with $f=p^{k-1}f_{0}+\bar{f_{1}}$, $f_{0}: V_{n}\rightarrow \mathbb{F}_{p}$, $\bar{f_{1}}: V_{n}\rightarrow \mathbb{Z}_{p^{k-1}}$, define the corresponding partition of $V_{n}$ as $\mathcal{P}_{f}=\{A(a), a \in \mathbb{Z}_{p^{k-1}}\}$, where $A(a)=\{x \in V_{n}: \bar{f_{1}}(x)=a\}$. Note that for any integer $k\geq 2$, by using Theorem 7, we can construct pairwise disjoint spectra generalized $t$-plateaued functions $F^{[j]}, j \in \mathbb{F}_{p}^{t}$, where $t\geq k-1$, for which the corresponding partition $\mathcal{P}_{F^{[j]}}$ has a large number of nonempty sets compared with the maximum number $p^{k-1}$.
\end{remark}

We give an example to illustrate that the above theorem can be used to construct pairwise disjoint spectra generalized plateaued functions whose Walsh supports are not affine subspaces.

\begin{example}\label{Example13}
Let $p=3, t=1, r=2, k=2$. Let $l_{c} \ (c \in \mathbb{F}_{3}): \mathbb{F}_{3}\times \mathbb{F}_{3}\rightarrow \mathbb{Z}_{3^2}$ be defined by
$l_{0}(x_{1}, x_{2})=3x_{1}^{2}$, $l_{1}(x_{1}, x_{2})=3(2x_{1}^{2}+x_{2})$, $l_{2}(x_{1}, x_{2})=3(x_{1}^{2}+2x_{2})$, $g_{0}, g_{1}: \mathbb{F}_{3}^{2} \times \mathbb{F}_{3}^{2}\rightarrow \mathbb{F}_{3}$ be defined by $g_{0}(y_{1}, \dots, y_{4})=y_{1}y_{3}+y_{2}y_{4}$, $g_{1}(y_{1}, \dots, y_{4})=y_{1}y_{3}+y_{2}y_{4}+2y_{3}y_{4}$, and $g^{[j]} \ (j \in \mathbb{F}_{3}): \mathbb{F}_{3}\rightarrow \mathbb{Z}_{3^2}$ be defined by $g^{[0]}(x)=x^{2}$, $g^{[1]}(x)=g^{[2]}(x)=x$. Define $f^{[j]}_{c}=l_{c+j}, c, j \in \mathbb{F}_{3}$. Then the generalized $1$-plateaued functions $F^{[j]} \ (j \in \mathbb{F}_{3}): \mathbb{F}_{3}^{6}\rightarrow \mathbb{Z}_{3^2}$ constructed by Corollary 4 are $F^{[0]}(x_{1}, x_{2}, y_{1}, \dots, y_{4})=3(2x_{1}^{2}y_{3}^{2}y_{4}^{2}+2x_{1}^{2}y_{3}y_{4}+x_{2}y_{3}y_{4}+x_{1}^{2}+y_{1}y_{3}+y_{2}y_{4})+(y_{3}y_{4} \ mod \ 3)^{2}$, $F^{[1]}(x_{1}, x_{2}, y_{1}, \dots, y_{4})=3(2x_{1}^{2}y_{3}^{2}y_{4}^{2}+x_{2}y_{3}y_{4}+2x_{1}^{2}+y_{1}y_{3}+y_{2}y_{4}+x_{2})+(y_{3}y_{4} \ mod \ 3)$, $F^{[2]}(x_{1}, x_{2}, y_{1}, \dots, y_{4})=3(2x_{1}^{2}y_{3}^{2}y_{4}^{2}+x_{1}^{2}y_{3}y_{4}+x_{2}y_{3}y_{4}+x_{1}^{2}+y_{1}y_{3}+y_{2}y_{4}+2x_{2})+(y_{3}y_{4} \ mod \ 3)$. By Theorem 7, $F^{[j]}, j\in \mathbb{F}_{3}$ are pairwise disjoint spectra generalized $1$-plateaued functions and one can verify that $S_{F^{[j]}}, j \in \mathbb{F}_{3}$ are all not affine subspaces.
\end{example}

By using pairwise disjoint spectra generalized plateaued functions as building blocks, we give the following construction of generalized bent functions, which is an extension of Theorem 2 of \cite{Cesmelioglu2}.

\begin{theorem}\label{Theorem8}
Let $p$ be a prime, $n, s \ (\leq n), k$ be positive integers and $n+s$ be even for $p=2, k=1$. Let $f_{y} \ (y \in \mathbb{F}_{p}^{s}): \mathbb{F}_{p}^{n}\rightarrow \mathbb{Z}_{p^k}$ be pairwise disjoint spectra generalized $s$-plateaued functions. Let $W$ and $U$ be $n$-dimensional and $s$-dimensional linear subspaces of $\mathbb{F}_{p}^{n+s}$ respectively and satisfy $\mathbb{F}_{p}^{n+s}=W\oplus U$. Define
\begin{equation*}
  F(xM+\pi(y))=f_{y}(x), x \in \mathbb{F}_{p}^{n}, y \in \mathbb{F}_{p}^{s},
\end{equation*}
where $M$ is a matrix whose row vectors form a basis of $W$ and $\pi$ is a bijection from $\mathbb{F}_{p}^{s}$ to $U$. Then $F$ is a generalized bent function from $\mathbb{F}_{p}^{n+s}$ to $\mathbb{Z}_{p^k}$.
\end{theorem}
\begin{proof}
First it is easy to see that $F$ is a function from $\mathbb{F}_{p}^{n+s}$ to $\mathbb{Z}_{p^k}$. For any $a \in \mathbb{F}_{p}^{n+s}$,
 \begin{equation*}
 \begin{split}
   W_{F}(a)& =\sum_{x \in \mathbb{F}_{p}^{n}}\sum_{y \in \mathbb{F}_{p}^{s}}\zeta_{p^k}^{f_{y}(x)}\zeta_{p}^{-a \cdot (xM+\pi(y))}\\
           & =\sum_{y \in \mathbb{F}_{p}^{s}}\zeta_{p}^{-a\cdot \pi(y)}W_{f_{y}}(aM^{T}).
 \end{split}
\end{equation*}
Since $f_{y}, y \in \mathbb{F}_{p}^{s}$ are pairwise disjoint spectra generalized $s$-plateaued functions, we have $|S_{f_{y}}|=p^{n-s}$ and $S_{f_{y}}\cap S_{f_{y'}}=\emptyset$ for any $y\neq y'$, which yields that $S_{f_{y}}, y \in \mathbb{F}_{p}^{s}$ is a partition of $\mathbb{F}_{p}^{n}$. Hence for any $a \in \mathbb{F}_{p}^{n+s}$, there exists a unique $y_{a}\in \mathbb{F}_{p}^{s}$ such that $aM^{T} \in S_{f_{y_{a}}}$ and $|W_{F}(a)|=|\zeta_{p}^{-a\cdot \pi(y_{a})}W_{f_{y_{a}}}(aM^{T})|=p^{\frac{n+s}{2}}$, that is, $F$ is a generalized bent function.
\end{proof}

When $k=1$, $W=\mathbb{F}_{p}^{n}\times \{0_{s}\}$, $U=\{0_{n}\}\times \mathbb{F}_{p}^{s}$, $M$ is the matrix whose row vectors are $(1, 0, \dots, 0, 0 , \dots, 0), (0, 1, \dots, 0, 0, \dots, 0), \dots, (0, 0, \dots, 1, 0, \dots, 0)$ and $\pi(y)=(0_{n}, y), y \in \mathbb{F}_{p}^{s}$, where $0_{n}$ denotes the zero vector of $\mathbb{F}_{p}^{n}$, Theorem 8 reduces to Theorem 2 of \cite{Cesmelioglu2}. We give two examples to illustrate Theorem 8.

\begin{example}\label{Example14}
Let $p=2, n=5, s=1, k=3$. Let $f_{0}, f_{1}: \mathbb{F}_{2}^{5}\rightarrow \mathbb{Z}_{2^3}$ be defined as $f_{0}(x_{1}, \dots, x_{5})=4(x_{1}x_{3}+x_{2}x_{4})+2x_{3}+x_{3}x_{4}$, $f_{1}(x_{1}, \dots, x_{5})=4(x_{1}x_{3}+x_{2}x_{4}+x_{5})+2x_{1}x_{2}+x_{1}$. Then $f_{0}, f_{1}$ are disjoint spectra generalized $1$-plateaued functions constructed in Example 11. Let $W=\mathbb{F}_{2}^{5} \times \{0\}$, $U=\{0_{5}\}\times \mathbb{F}_{2}$, $M$ is the matrix whose row vectors are $(1, 0, \dots, 0, 0), \dots, (0, 0, \dots, 1, 0)$ and $\pi(y)=(0, \dots, 0, y), y \in \mathbb{F}_{2}$. Then the constructed generalized bent function $F: \mathbb{F}_{2}^{6}\rightarrow \mathbb{Z}_{2^3}$ by Theorem 8 is $F(x_{1}, \dots, x_{6})=f_{x_{6}}(x_{1}, \dots, x_{5})=4(x_{1}x_{3}+x_{2}x_{4}+x_{5}x_{6})+2((x_{1}x_{2}x_{6}+x_{3}(1+x_{6})) \ mod \ 2)+((x_{3}x_{4}(1+x_{6})+x_{1}x_{6}) \ mod \ 2)$.
\end{example}

\begin{example}\label{Example15}
Let $p=3, n=6, s=1, k=2$. Let $F^{[j]} \ (j \in \mathbb{F}_{3}): \mathbb{F}_{3}^{6}\rightarrow \mathbb{Z}_{3^2}$ be the pairwise disjoint spectra generalized $1$-plateaued functions constructed in Example 13. Let $W=\mathbb{F}_{3}^{6} \times \{0\}$, $U=\{0_{6}\}\times \mathbb{F}_{3}$, $M$ be the matrix whose row vectors are $(1, 0, \dots, 0, 0), \dots, (0, 0, \dots, 1, 0)$ and $\pi(y)=(0_{6}, y), y \in \mathbb{F}_{3}$. Then by Theorem 8, the function $F(x, y)=F^{[y]}(x), x \in \mathbb{F}_{3}^{6}, y\in \mathbb{F}_{3}$ is a generalized bent function.
\end{example}

\section{Conclusion}
In \cite{Mesnager9}, Mesnager \emph{et al.} introduced generalized plateaued functions from $V_{n}$ to $\mathbb{Z}_{p^k}$ in order to study plateaued functions in the general context of generalized $p$-ary functions. The objective of this paper is to increase constructions of generalized $p$-ary plateaued functions for any prime $p$ as there are lacks of constructions of generalized $p$-ary plateaued functions with $p$ taking any prime. In particular, when $k=1$, the constructions in this paper are applicable for plateaued functions.

(1) By Theorems 1, 2, 3, one can construct generalized $p$-ary plateaued functions with (non)-affine Walsh supports by using known generalized $p$-ary bent functions as building blocks (Note that Theorem 4 is only for $p$-ary plateaued functions by using vectorial bent functions as building blocks). In particular, when $p=2, k=1$, the constructions in Theorems 2, 3, 4 provide an answer to Open Problem 2 proposed in \cite{Hodzic3}.

(2) By Corollaries 4, 5, 6, one can construct generalized $p$-ary bent functions and generalized $p$-ary $s$-plateaued functions with larger variables by using generalized $p$-ary bent functions, respectively, generalized $p$-ary $s$-plateaued functions as building blocks. In particular, we show that the canonical way to construct the so-called Generalized Maiorana-McFarland bent functions can be obtained by Corollary 4 and we illustrate that Corollary 5 can be used to construct bent functions not in the completed Generalized Maiorana-McFarland class (see Example 8 and Appendix). Based on Corollary 5, we also give constructions of plateaued functions in the subclass \emph{WRP} of the class of weakly regular plateaued functions and vectorial plateaued functions.

(3) By the discussions in Section V, one can construct pairwise disjoint spectra generalized $p$-ary plateaued functions by Theorems 1, 2, 3, 7 (Note that Theorem 4 can only be used to construct pairwise disjoint spectra $p$-ary plateaued functions). By using pairwise disjoint spectra generalized $p$-ary plateaued functions as building blocks, one can construct generalized $p$-ary bent functions by Theorem 8.

Therefore, by recursively using our constructions, one can obtain infinitely many generalized $p$-ary plateaued functions with (non)-affine Walsh supports for any prime $p$.

Plateaued functions have important applications in coding theory, sequences and combinatorics. For examples, Mesnager \emph{et al}. \cite{Mesnager3} presented constructions of linear codes from weakly regular plateaued functions and the secret sharing schemes based on these linear codes. Mesnager and S{\i}nak \cite{Mesnager7,Mesnager8} constructed several classes of minimal linear codes with few weights and strongly regular graphs, association schemes from weakly regular plateaued functions. Bozta\c{s} \emph{et al}. \cite{Serdar} used plateaued functions to design sequences with good correlation properties. It is interesting to further study the applications of generalized plateaued functions in coding theory, sequences and combinatorics. For examples, constructing linear codes, sequences, strongly regular graphs and association schemes from generalized plateaued functions.

\appendix
We prove that the bent function constructed in Example 8 is not in the completed Generalized Maiorana-McFarland class.

Recall that the bent function constructed in Example 8 is $F(x, y_{1}, y_{2})=f_{g_{0}(y_{1}, y_{2})-g_{1}(y_{1}, y_{2})}(x)+g_{0}(y_{1}, y_{2})=f_{0}(x)+g_{0}(y_{1}, y_{2})+(g_{0}(y_{1}, y_{2})-g_{1}(y_{1}, y_{2}))^{2}(-f_{0}(x)-f_{1}(x)-f_{2}(x))+(g_{0}(y_{1}, y_{2})-g_{1}(y_{1}, y_{2}))(2f_{1}(x)+f_{2}(x)), (x, y_{1}, y_{2})\in \mathbb{F}_{3^4}\times \mathbb{F}_{3^2}\times \mathbb{F}_{3^2}$, where $f_{0}(x)=Tr_{1}^{4}(x^{34}+x^{2})$, $f_{1}(x)=Tr_{1}^{4}(x^{2})$, $f_{2}(x)=Tr_{1}^{4}(\xi x^{2})$, $g_{0}(y_{1}, y_{2})=Tr_{1}^{2}(y_{1}y_{2}^{7})$, $g_{1}(y_{1}, y_{2})=Tr_{1}^{2}(zy_{1}y_{2}^{7})$ and $\xi$ is the primitive element of $\mathbb{F}_{3^4}$ with $\xi^{4}+2\xi^{3}+2=0$, $z$ is the primitive element of $\mathbb{F}_{3^2}$ with $z^2+z+2=0$.

By Theorem 2 of \cite{Cesmelioglu3}, if $F$ is in the completed Generalized Maiorana-McFarland class, then for an integer $1\leq s \leq 4$ there exists an $s$-dimensional subspace $V$ of $\mathbb{F}_{3^4}\times \mathbb{F}_{3^2}\times \mathbb{F}_{3^2}$ such that the second order derivative
\begin{equation}\label{33}
  D_{a}D_{c}F(x, y_{1}, y_{2})=0
\end{equation}
for any $a=(a_{0}, a_{1}, a_{2}), c=(c_{0}, c_{1}, c_{2}) \in V, (x, y_{1}, y_{2})\in \mathbb{F}_{3^4}\times \mathbb{F}_{3^2}\times \mathbb{F}_{3^2}$. Define $\bar{g_{i}}(y)=g_{i}(y_{1}, y_{2}), i=0,1$ and $\bar{F}(x, y)=f_{\bar{g_{0}}(y)-\bar{g_{1}}(y)}(x)+\bar{g_{0}}(y)$, where $y=(y_{1, 1}, y_{1, 2}, y_{2, 1}, y_{2, 2})\in \mathbb{F}_{3}^{4}, (y_{1}, y_{2})\in \mathbb{F}_{3^2}\times \mathbb{F}_{3^2}, y_{1}=y_{1, 1}+y_{1, 2}z, y_{2}=y_{2, 1}+y_{2, 2}z$. Then $\bar{F}$ is a non-weakly regular bent function from $\mathbb{F}_{3^4}\times \mathbb{F}_{3}^{4}$ to $\mathbb{F}_{3}$. By simple calculation we have $\bar{g_{0}}(y)-\bar{g_{1}}(y)=(y_{1, 1}+y_{1, 2})y_{2, 1}^{2}y_{2, 2}+(2y_{1, 1}+y_{1, 2})y_{2, 1}y_{2, 2}^{2}+2y_{1, 1}y_{2, 2}+2y_{1, 2}y_{2, 1}$, $(\bar{g_{0}}(y)-\bar{g_{1}}(y))^{2}=y_{1, 1}^{2}y_{2, 2}^{2}+y_{1, 2}^{2}y_{2, 1}^{2}+y_{1, 1}y_{1, 2}y_{2, 1}y_{2, 2}$, where $y=(y_{1, 1}, y_{1, 2}, y_{2, 1}, y_{2, 2})\in \mathbb{F}_{3}^{4}$.

Suppose (33) holds. Then
\begin{equation}\label{34}
  D_{\bar{a}}D_{\bar{c}}\bar{F}(x, y)=0
\end{equation}
for any $a=(a_{0}, a_{1, 1}, a_{1, 2}, a_{2, 1}, a_{2, 2}), c=(c_{0}, c_{1, 1}, c_{1, 2}, c_{2, 1}, c_{2, 2}) \in \bar{V}, (x, y)\in \mathbb{F}_{3^4}\times \mathbb{F}_{3}^{4}$, where $\bar{V}=\{(a_{0}, a_{1, 1}, a_{1, 2}, a_{2, 1}, a_{2, 2})\in \mathbb{F}_{3^4}\times \mathbb{F}_{3}^{4}: (a_{0}, a_{1, 1}+a_{1, 2}z, a_{2, 1}+a_{2, 2}z) \in V \}, y=(y_{1, 1}, y_{1, 2},$ $ y_{2, 1}, y_{2, 2})\in \mathbb{F}_{3}^{4}$.
As $\{30 \cdot 3^{i} \ (mod \ (3^4-1)): i\geq 0\}=\{10, 30\}$ and $\left(\begin{array}{c}
                                                                         34 \\
                                                                         10
                                                                       \end{array}
\right)\equiv 0 \ (mod \ 3)$, $\left(\begin{array}{c}
                                                                         34 \\
                                                                         30
                                                                       \end{array}
\right)\equiv 2 \ (mod \ 3)$, $D_{\bar{a}}D_{\bar{c}}\bar{F}$ contains $-y_{1, 1}^{2}y_{2, 2}^{2}Tr_{1}^{4}(2((a_{0}+c_{0})^{4}-a_{0}^{4}-c_{0}^{4})x^{30})$. Then by (34), $Tr_{2}^{4}((a_{0}+c_{0})^{4}-a_{0}^{4}-c_{0}^{4})=0$ for any $\bar{a}=(a_{0}, a_{1, 1}, a_{1, 2}, a_{2, 1}, a_{2, 2}), \bar{c}=(c_{0}, c_{1, 1}, c_{1, 2}, c_{2, 1}, c_{2, 2})\in \bar{V}$. One can verify that for $a \in \mathbb{F}_{3^4}$, $Tr_{2}^{4}(a^{4})=0$ if and only if $a=0$. If there exists $a_{0}\neq 0$ such that $\bar{a}=(a_{0}, a_{1, 1}, a_{1, 2}, a_{2, 1}, a_{2, 2})\in \bar{V}$, let $\bar{c}=\bar{a}$, then $c_{0}=a_{0}\neq 0$ and $Tr_{2}^{4}((a_{0}+c_{0})^{4}-a_{0}^{4}-c_{0}^{4})=Tr_{2}^{4}(2a_{0}^{4})\neq 0$, which is a contradiction. Hence $\bar{V}\subseteq \{0\} \times \mathbb{F}_{3}^{4}$, that is, $V \subseteq \{0\} \times \mathbb{F}_{3^2}\times \mathbb{F}_{3^2}$. For any fixed $(0, a_{1}, a_{2})$, $(0, c_{1}, c_{2})\in V$ and $(y_{1}, y_{2})\in \mathbb{F}_{3^2}\times \mathbb{F}_{3^2}$, let $d_{0}=D_{(a_{1}, a_{2})}D_{(c_{1}, c_{2})}g_{0}(y_{1}, y_{2})$, $d_{1}=D_{(a_{1}, a_{2})}D_{(c_{1}, c_{2})}(g_{0}(y_{1}, y_{2})-g_{1}(y_{1}, y_{2}))$, $d_{2}=D_{(a_{1}, a_{2})}D_{(c_{1}, c_{2})}(g_{0}(y_{1}, y_{2})-g_{1}(y_{1}, y_{2}))^{2}$. By $D_{(0, a_{1}, a_{2})}D_{(0, c_{1}, c_{2})}F(x, y_{1}, y_{2})=D_{(a_{1}, a_{2})}D_{(c_{1}, c_{2})}g_{0}(y_{1}, y_{2})+(-f_{0}(x)-f_{1}(x)-f_{2}(x))D_{(a_{1}, a_{2})}D_{(c_{1}, c_{2})}$ $(g_{0}(y_{1}, y_{2})-g_{1}(y_{1}, y_{2}))^{2}+(2f_{1}(x)+f_{2}(x))D_{(a_{1}, a_{2})}D_{(c_{1}, c_{2})}(g_{0}(y_{1}, y_{2})-g_{1}(y_{1}, y_{2}))=0$ for any $(0, a_{1}, a_{2}), (0, c_{1}, c_{2})\in V, (x, y_{1}, y_{2})\in \mathbb{F}_{3^4}\times \mathbb{F}_{3^2}\times \mathbb{F}_{3^2}$, for any fixed $(0, a_{1}, a_{2}), (0, c_{1}, c_{2})\in V$ and $(y_{1}, y_{2})\in \mathbb{F}_{3^2}\times \mathbb{F}_{3^2}$, we have $-d_{2}f_{0}(x)+(2d_{1}-d_{2})f_{1}(x)+(d_{1}-d_{2})f_{2}(x)=-d_{0}, x \in \mathbb{F}_{3^4}$. By $f_{0}(0)=f_{1}(0)=f_{2}(0)=0$, we have $d_{0}=0$.
By $i+j\xi\neq 0$ for any $i, j\in \mathbb{F}_{3}$ and the algebraic degree of $f_{0}$ is 4, the algebraic degree of $f_{1}$ and $f_{2}$ is 2, we have $f_{0}, f_{1}, f_{2}$ are linearly independent, hence $d_{1}=d_{2}=0$. Therefore, (33) holds if and only if for any $(0, a_{1}, a_{2}), (0, c_{1}, c_{2})\in V, (y_{1}, y_{2})\in \mathbb{F}_{3^2}\times \mathbb{F}_{3^2}$,
\begin{equation}\label{35}
  D_{(a_{1}, a_{2})}D_{(c_{1}, c_{2})}g_{0}(y_{1}, y_{2})=0
\end{equation}
and
\begin{equation}\label{36}
D_{(a_{1}, a_{2})}D_{(c_{1}, c_{2})}(g_{0}(y_{1}, y_{2})-g_{1}(y_{1}, y_{2}))=0
\end{equation}
and
\begin{equation}\label{37}
  D_{(a_{1}, a_{2})}D_{(c_{1}, c_{2})}(g_{0}(y_{1}, y_{2})-g_{1}(y_{1}, y_{2}))^{2}=0.
\end{equation}
By (35), (36) and the fact that $\{1, 1-z\}$ is a basis of $\mathbb{F}_{3^2}$ over $\mathbb{F}_{3}$, we have for any fixed $(0, a_{1}, a_{2}), (0, c_{1}, c_{2})\in V$ and $(y_{1}, y_{2})\in \mathbb{F}_{3^2}\times \mathbb{F}_{3^2}$, $Tr_{1}^{2}(((y_{1}+a_{1}+c_{1})(y_{2}+a_{2}+c_{2})^{7}-(y_{1}+a_{1})(y_{2}+a_{2})^{7}-(y_{1}+c_{1})(y_{2}+c_{2})^{7}+y_{1}y_{2}^{7}) x)=0, x \in \mathbb{F}_{3^2}$, which yields $(y_{1}+a_{1}+c_{1})(y_{2}+a_{2}+c_{2})^{7}-(y_{1}+a_{1})(y_{2}+a_{2})^{7}-(y_{1}+c_{1})(y_{2}+c_{2})^{7}+y_{1}y_{2}^{7}=0$ for any $(0, a_{1}, a_{2}), (0, c_{1}, c_{2})\in V$ and $(y_{1}, y_{2})\in \mathbb{F}_{3^2}\times \mathbb{F}_{3^2}$. We claim $V\subseteq \{0\} \times \mathbb{F}_{3^2}\times \{0\}$. If there exists $a_{2}\neq 0 $ such that $a=(0, a_{1}, a_{2})\in V$, let $c=a$. Then $c_{2}=a_{2}\neq 0$ and the coefficient of $y_{1}y_{2}^{3}$ is $\left(\begin{array}{c}
                                                                         7 \\
                                                                         3
                                                                       \end{array}\right)((a_{2}+c_{2})^{4}-a_{2}^{4}-c_{2}^{4})=a_{2}^{4}\neq 0$, which is a contradiction. Hence $V\subseteq \{0\} \times \mathbb{F}_{3^2}\times \{0\}$, that is, $\bar{V} \subseteq \{0\} \times \mathbb{F}_{3}^{2} \times \{(0, 0)\}$. By (37), we have
$D_{(a_{1, 1}, a_{1, 2}, 0, 0)}D_{(c_{1, 1}, c_{1, 2}, 0, 0)}(\bar{g_{0}}(y)-\bar{g_{1}}(y))^{2}=0$ for any $(0, a_{1, 1}, a_{1, 2}, 0, 0), (0, c_{1, 1}, c_{1, 2}, 0, 0)\in \bar{V}, y=(y_{1, 1}, y_{1, 2}, y_{2, 1}, y_{2, 2})\in \mathbb{F}_{3}^{4}$. By simple calculation, we have $2a_{1, 1}c_{1, 1}y_{2, 2}^{2}+2a_{1, 2}c_{1, 2}y_{2, 1}^{2}+(a_{1, 1}c_{1, 2}+a_{1, 2}c_{1, 1})y_{2, 1}y_{2, 2}=0$, which yields $a_{1, 1}c_{1, 1}=a_{1, 2}c_{1, 2}=a_{1, 1}c_{1, 2}+a_{1, 2}c_{1, 1}=0$ for any $(0, a_{1, 1}, a_{1, 2}, 0, 0), (0, c_{1, 1}, c_{1, 2}, 0, 0)\in \bar{V}$. If there exists $(a_{1, 1}, a_{1, 2})\neq (0, 0)$ such that $\bar{a}=(0, a_{1, 1}, a_{1, 2}, 0, 0)\in \bar{V}$, let $\bar{c}=\bar{a}$, then $a_{1, 1}c_{1, 1}=a_{1, 1}^{2}\neq 0$ or $a_{1, 2}c_{1, 2}=a_{1, 2}^{2}\neq 0$ since $(a_{1, 1}, a_{1, 2})\neq (0, 0)$, which is a contradiction. Hence, $\bar{V}=\{(0, 0, 0, 0, 0)\}$, that is, $V=\{(0, 0, 0)\}$. By Theorem 2 of \cite{Cesmelioglu3}, $F$ is not in the completed Generalized Maiorana-McFarland class.


\begin{thebibliography}{}
\bibitem{Serdar}
S. Bozta\c{s}, F. \"{O}zbudak and E. Tekin, Explicit full correlation distribution of sequence families using plateaued functions, \emph{IEEE Trans. Inf. Theory}, vol. 64, no. 4, pp. 2858-2875, 2018.
\bibitem{Carlet1}
C. Carlet, Boolean and vectorial plateaued functions and APN functions, \emph{IEEE Trans. Inf. Theory}, vol. 61, no. 11, pp. 6272-6289, 2015.
\bibitem{Carlet2}
C. Carlet, Boolean Functions for Cryptography and Coding Theory, Cambridge University Press, Cambridge 2020.
\bibitem{Carlet3}
C. Carlet, On the secondary constructions of resilient and bent functions, In Proceedings of the Workshop on Coding, Cryptography and Combinatorics 2003, published by Birkh\"{a}user Verlag, pp. 3-28, 2004.
\bibitem{Carlet4}
C. Carlet, Partially-bent functions, \emph{Des. Codes Cryptogr.}, vol. 3, no. 2, pp. 135-145, 1993.
\bibitem{Carlet5}
C. Carlet and S. Mesnager, Four decades of research on bent functions, \emph{Des. Codes Cryptogr.}, vol. 78, no. 1, pp. 5-50, 2016.
\bibitem{Cesmelioglu1}
A. \c{C}e\c{s}melio\v{g}lu, G. McGuire and W. Meidl, A construction of weakly and non-weakly regular bent functions, \emph{J. Combinat. Theory A}, vol. 119, no. 2, pp. 420-429, 2012.
\bibitem{Cesmelioglu2}
A. \c{C}e\c{s}melio\v{g}lu and W. Meidl, A construction of bent functions from plateaued functions, \emph{Des. Codes Cryptogr.}, vol. 66, nos. 1-3, pp. 231-242, 2013.
\bibitem{Cesmelioglu3}
A. \c{C}e\c{s}melio\v{g}lu, W. Meidl and A. Pott, Generalized Maiorana-McFarland class and normality of $p$-ary bent functions, \emph{Finite Fields Appl.}, vol. 24, pp. 105-117, 2013.
\bibitem{Dillon}
J. F. Dillon, Elementary Hadamard difference sets, Ph. D. Thesis, University of Maryland, 1974.
\bibitem{Helleseth}
T. Helleseth, A. Kholosha, Monomial and quadratic bent functions over the finite fields of odd characteristic, \emph{IEEE Trans. Inf. Theory}, vol. 52, no. 5, pp. 2018-2032, 2006.
\bibitem{Hodzic1}
S. Hod\v{z}i\'{c}, W. Meidl and E. Pasalic, Full characterization of generalized bent functions as (semi)-bent spaces, their dual, and the Gray image, \emph{IEEE Trans. Inf. Theory}, vol. 64, no. 7, pp. 5432-5440, 2018.
\bibitem{Hodzic2}
S. Hod\v{z}i\'{c} and E. Pasalic, Construction methods for generalized bent functions, \emph{Discret. Appl. Math.}, vol. 238, pp. 14-23, 2018.
\bibitem{Hodzic3}
S. Hod\v{z}i\'{c}, E. Pasalic, Y. Wei and F. Zhang, Designing plateaued Boolean functions in spectral domain and their classification, \emph{IEEE Trans. Inf. Theory}, vol. 65, no. 9, pp. 5865-5879, 2019.
\bibitem{Hyun}
J. Y. Hyun, J. Lee and Y. Lee, Explicit criteria for construction of plateaued functions, \emph{IEEE Trans. Inf. Theory}, vol. 62, no. 12, pp. 7555-7565, 2016.
\bibitem{Kumar}
P. V. Kumar, R. A. Scholtz and L. R. Welch, Generalized bent functions and their properties, \emph{J. Combinat. Theory A}, vol. 40, no. 1, pp. 90-107, 1985.
\bibitem{Lisonek}
P. Lison\v{e}k and H. Y. Lu, Bent functions on partial spreads, \emph{Des. Codes Cryptogr.}, vol. 73, no. 1, pp. 209-216, 2014.
\bibitem{Martinsen1}
T. Martinsen, W. Meidl, S. Mesnager and P. St\u{a}nic\u{a}, Decomposing generalized bent and hyperbent functions, \emph{IEEE Trans. Inf. Theory}, vol. 63, no. 12, pp. 7804-7812, 2017.
\bibitem{Martinsen2}
T. Martinsen, W. Meidl and P. St\u{a}nic\u{a}, Partial spread and vectorial generalized bent functions, \emph{Des. Codes Cryptogr.}, vol. 85, no. 1, pp. 1-13, 2017.
\bibitem{Meidl}
W. Meidl and A. Pott, Generalized bent functions into $\mathbb{Z}_{p^k}$ from the partial spread and the Maiorana-McFarland class, \emph{Cryptogr. Commun.}, vol. 11, no. 6, pp. 1233-1245, 2019.
\bibitem{Mesnager1}
S. Mesnager, Bent Functions-Fundamentals and Results, Springer, Switzerland, 2016.
\bibitem{Mesnager2}
S. Mesnager, On generalized hyper-bent functions, \emph{Cryptogr. Commun.}, vol. 12, no. 3, pp. 455-468, 2020.
\bibitem{Mesnager3}
S. Mesnager, F. \"{O}zbudak and A. S{\i}nak, Linear codes from weakly regular plateaued functions and their secret sharing schemes, \emph{Des. Codes Cryptogr.}, vol. 87, nos. 2-3, pp. 463-480, 2019.
\bibitem{Mesnager4}
S. Mesnager, F. \"{O}zbudak and A. S{\i}nak, On the $p$-ary (cubic) bent and plateaued (vectorial) functions, \emph{Des. Codes Cryptogr.}, vol. 86, no. 8, pp. 1865-1892, 2018.
\bibitem{Mesnager5}
S. Mesnager, F. \"{O}zbudak and A. S{\i}nak, Results on characterizations of plateaued functions in arbitrary characteristic, Cryptography and information security in the Balkans, BalkanCrytSec 2015, Koper, Slovenia, Revised Selected Papers. In: Pasalic E., Knudsen L.R.(eds.) LNCS 9540, pp. 17-30, Springer, Berlin, 2016.
\bibitem{Mesnager6}
S. Mesnager, C. Riera and P. St\u{a}nic\u{a}, Multiple characters transforms and generalized Boolean functions, \emph{Cryptogr. Commun.}, vol. 11, no. 6, pp. 1247-1260, 2019.
\bibitem{Mesnager7}
S. Mesnager and A. S{\i}nak, Several classes of minimal linear codes with few weights from weakly regular plateaued functions, \emph{IEEE Trans. Inf. Theory}, vol. 66, no. 4, pp. 2296-2310, 2020.
\bibitem{Mesnager8}
S. Mesnager and A. S{\i}nak, Strongly regular graphs from weakly regular plateaued functions, 2019 Ninth International Workshop on Signal Design and its Applications in Communications (IWSDA), Dongguan, China, pp. 1-5, 2019.
\bibitem{Mesnager9}
S. Mesnager, C. Tang and Y. Qi, Generalized plateaued functions and admissible (plateaued) functions, \emph{IEEE Trans. Inf. Theory}, vol. 63, no. 10, pp. 6139-6148, 2017.
\bibitem{Mesnager10}
S. Mesnager, C. Tang, Y. Qi, L. Wang, B. Wu and K. Feng, Further results on generalized bent functions and their complete characterization, \emph{IEEE Trans. Inf. Theory}, vol. 64, no. 7, pp. 5441-5452, 2018.
\bibitem{Oktay}
O. Olmez, Plateaued functions and one-and-half difference sets, \emph{Des. Codes Cryptogr.}, vol. 76, no. 3, pp. 537-549, 2015.
\bibitem{Potapov}
V. N. Potapov, On $q$-ary bent and plateaued functions, \emph{Des. Codes Cryptogr.}, vol. 88, no. 10, pp. 2037-2049, 2020.
\bibitem{Qi}
Y. Qi, C. Tang, Z. Zhou and C. Fan, Several infinite families of $p$-ary weakly regular bent functions, \emph{Adv. Math. Commun.}, vol. 12, no. 2, pp. 303-315, 2018.
\bibitem{Riera}
C. Riera and P. St\u{a}nic\u{a}, Landscape Boolean functions, \emph{Adv. Math. Commun.}, vol. 13, no. 4, pp. 613-627, 2019.
\bibitem{Rothaus}
O. S. Rothaus, On ``bent'' functions, \emph{J. Combinat. Theory A}, vol. 20, no. 3, pp. 300-305, 1976.
\bibitem{Stanica}
P. St\u{a}nic\u{a}, T. Martinsen, S. Gangopadhyay and B. K. Singh, Bent and generalized bent Boolean functions, \emph{Des. Codes Cryptogr.}, vol. 69, no. 1, pp. 77-94, 2013.
\bibitem{Tang}
C. Tang, C. Xiang, Y. Qi and K. Feng, Complete characterization of generalized bent and $2^{k}$-bent Boolean functions, \emph{IEEE Trans. Inf. Theory}, vol. 63, no. 7, pp. 4668-4674, 2017.
\bibitem{Zheng}
Y. Zheng and X. M. Zhang, On plateaued functions, \emph{IEEE Trans. Inf. Theory}, vol. 47, no. 3, pp. 1215-1223, 2001.



\end{thebibliography}
\end{document}